\documentclass[12pt]{report}

\usepackage{amssymb}
\usepackage{amsmath}
\usepackage{amsthm}
\usepackage{stmaryrd}
\usepackage{graphicx}
\usepackage{tikz}
\usepackage{geometry}
\usepackage{verbatim}
\usepackage{caption}
\usepackage{subcaption}

\usetikzlibrary{shapes}

\tikzstyle{noir}=[circle,draw,fill=black!40]
\tikzstyle{debut}=[circle,draw,very thick,fill=white]
\tikzstyle{blanc}=[circle,draw, fill=white]
\title{Towards an algebraic theory of rational word functions}
\author{Nathan Lhote}

\newtheorem{thm}{Theorem}[section]
\newtheorem{prp}{Proposition}[section]
\newtheorem{lem}{Lemma}[section]
\newtheorem{cor}{Corollary}[section]

\theoremstyle{definition}

\newtheorem{xmp}{Example}[section]

\theoremstyle{remark}
\newtheorem{rmk}{Remark}[section]

\newenvironment{prf}
{\textit{Proof:}}
{\hfill $\square$\\}

\newcommand{\HRule}{\rule{\linewidth}{0.5mm}}

\begin{document}

\begin{titlepage}

\begin{center}

\includegraphics[width=0.4\textwidth]{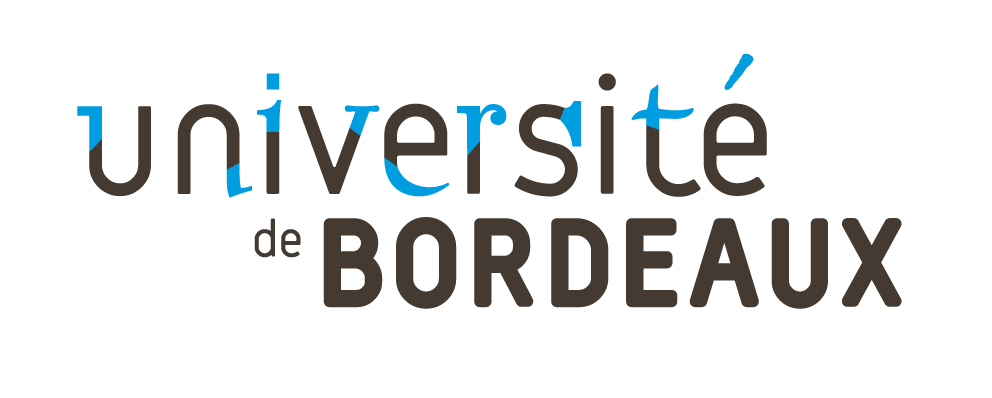}~\hfill ~\includegraphics[width=0.4\textwidth]{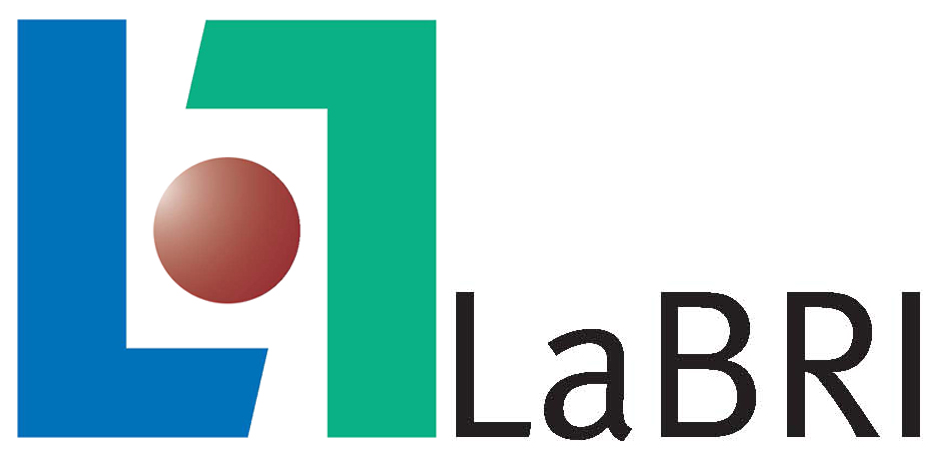}~\\[1cm]

\Large Internship report\\[0.5cm]

\HRule \\[0.4cm]
{ \huge Towards an algebraic characterization of rational word functions \\[0.4cm] }

\HRule \\[1.5cm]

\noindent
\begin{minipage}[t]{0.4\textwidth}
\begin{flushleft} \large
\emph{Author:}\\
Nathan Lhote
\end{flushleft}
\end{minipage}
\begin{minipage}[t]{0.4\textwidth}
\begin{flushright} \large
\emph{Supervisors:} \\
Emmanuel Filiot\footnotemark[1]\\
Olivier Gauwin\footnotemark[2]\\
Anca Muscholl\footnotemark[2]
\end{flushright}
\end{minipage}

\vfill

{\large \today}

\end{center}
\footnotetext[1]{Universit\'e Libre de Bruxelles}
\footnotetext[2]{LaBRI, Universit\'e de Bordeaux}

\end{titlepage}

\chapter*{Introduction}

In formal language theory, several models characterize regular languages: finite automata, congruences of finite index, monadic second-order logic (MSO), see \emph{e.g.}~\cite{buchi60,elgot61, trakhtenbrot61}.
Further connections exist between monoid varieties and logical fragments of MSO, for instance aperiodic monoids have been shown to recognize exactly the word languages defined by first-order (FO) formulas \cite{schutzenberger65,mcnaughtonp71}.

When considering word relations instead of languages, automata are replaced by transducers, that is automata with outputs associated with transitions. A transducer is an automaton that reads an input word and returns an output word obtained by concatenating the outputs of the transitions. However, several important properties do not generalize from automata to transducers. For instance, the well known equivalence between deterministic and non-deterministic one-way automata, as well as the equivalence between one-way and two-way automata, do not transfer to transducers. One property that is preserved is the equivalence between automata and MSO formulas: it has been shown \cite{engelfrieth01} that MSO word transductions and two-way transducers define the same class of word functions called \emph{regular} functions. A recently introduced model of computation, streaming string transducers (SST), has been shown to compute the same class of regular functions \cite{alurc10}.
Two recent related results are the equivalence between FO transductions and aperiodic SSTs \cite{filiotkt14} and the equivalence between FO transductions and aperiodic two-way transducers \cite{cartond15}. However it is not known if one can decide if a given regular function is FO definable.

\paragraph*{Our contribution:}
Our result deals with functions that are definable by one-way word transducers.
These functions are known in the literature as \emph{rational} functions.

The notion of minimal automaton goes beyond minimizing the state space.
Indeed to decide whether a regular language satisfies some algebraic property, like aperiodicity, it suffices to consider the minimal automaton.
Therefore in order to have an algebraic characterization of rational functions, we need a notion similar to the one of minimal automata for transductions.
For the class of functions defined by deterministic transducers, such a notion exists \cite{choffrut03} and this minimal transducer enjoys, among deterministic transducers, the same kind of minimality properties.

In an attempt to obtain a similar notion for rational functions, we study the model of bimachines \cite{schutzenberger61} which has been shown to be a canonical model for rational functions (see \emph{e.g.} \cite{berstelb79}).
We describe a canonical bimachine, introduced by \cite{reutenauers91}, and show that this representation preserves, similarly to the minimal automaton for languages, some algebraic properties of the function.
In the case of aperiodic rational function, \emph{i.e.} functions definable by an aperiodic one-way transducer, a characterization has already been given in \cite{reutenauers95}.
Our main contribution is to give an effective characterization of aperiodic rational functions, and extend it to other algebraic varieties of rational functions.
We also introduce \emph{translations}, a model of logical transductions inspired by and equivalent to bimachines.
We show an equivalence between transducers satisfying some algebraic property and logical translations, for instance between aperiodic transducers and FO-translations.

\tableofcontents

\chapter{Automata for languages and transductions}

\section{Words, languages and transductions}

Let an \emph{alphabet} $\Sigma$ be a finite set of symbols called \emph{letters}.
Let $\Sigma^\ast$ denote the free monoid on $\Sigma$.
An element of $\Sigma^\ast$ is called a \emph{word} and the identity element $\varepsilon$ is called the \emph{empty word}.
A \emph{language} is a set of words.

A \emph{transduction} $R$ is a subset of $\Sigma^\ast\times\Sigma^\ast$.
For $(u,v)\in R$, $u$ is called an \emph{input} word and $v$ an \emph{output word}.
The \emph{domain} of $R$, denoted by $\mathrm{dom}(R)$, is the set of input words.
A transduction is called \emph{functional} if it is a (partial) function.
For a functional transduction $f$ we write $f(u)=v$ instead of $(u,v)\in f$.
\begin{xmp}
We give two running examples of functional transductions:

\begin{itemize}

\item Let $\Sigma=\left\{a\right\}$.
The transduction $f_{even}:\Sigma^\ast\rightarrow \Sigma^\ast$ copies the word if its length is even and outputs $\varepsilon$ otherwise.

\item Let $\Sigma=\left\{a,b\right\}$.
The transduction $f_{ends}:\Sigma^\ast\rightarrow \Sigma^\ast$ replaces each letter by $a$ if the first and last letter are $a$ and yields $\varepsilon$ otherwise. \emph{e.g.}, $f_{ends}(abaa)=a^4$, $f_{ends}(baaab)=f_{ends}(abab)=\varepsilon$.

\end{itemize}

\end{xmp}

\section{Finite state automata}

A \emph{left non-deterministic finite state automaton} (NFA or simply automaton) over $\Sigma$ is a tuple $A=\left(Q,I,F,\Delta \right)$,
where $Q$ is a finite set of states,
$I\subseteq Q$ is the set of initial states,
$F\subseteq Q$ is the set of accepting states,
$\Delta\subseteq Q\times\Sigma\times Q$ is the transition relation.

We denote $(p,\sigma,q)\in \Delta$ by $p\xrightarrow{\sigma}_A q$ (we will omit the $A$ when the context is clear).
We extend the relation to words: $p_0\xrightarrow{\sigma_1\ldots\sigma_k} p_k$ when we have $p_0\xrightarrow{\sigma_1} p_1\xrightarrow{\sigma_2} p_2\cdots p_{k-1}\xrightarrow{\sigma_k} p_k$. Such a sequence is called a \emph{run of $A$ on $\sigma_1\ldots\sigma_k$}, and if $p_0\in I$ and $p_k\in F$, it is called a \emph{successful run}. We also add $q\xrightarrow \varepsilon q$, $\forall q\in Q$.

A word $u\in \Sigma^\ast$ is \emph{recognized} or \emph{accepted} by $A$ if there exist $p\in I$, $q\in F$ such that $p\xrightarrow u q$. The set of all words recognized by $A$ is called the \emph{language recognized} by $A$ and is denoted by $L(A)$.

A \emph{left deterministic finite state automaton} (left DFA or simply DFA) on an alphabet $\Sigma$ is defined similarly to an NFA as a tuple $A=(Q,q_0,F,\delta)$ where the transition function is $\delta: Q\times\Sigma\rightarrow Q$ and $q_0\in Q$ is the unique initial state.

Let $A=\left(Q,I,F,\Delta \right)$ be an automaton, and let $P_1,P_2\subseteq Q$. Then we define $L_{P_1,P_2}(A)$ as the language recognized by $A'=\left(Q,P_1,P_2,\Delta \right)$.

\begin{rmk}
Right automata are defined exactly as left automata except that they read the words from right to left.
Let $A=\left(Q,I,F,\Delta \right)$ be a right automaton.
Then we denote $(p,\sigma,q)\in \Delta$ by $q\xleftarrow \sigma p$ and
$p_k\xleftarrow{\sigma_1\ldots\sigma_k} p_0$ when we have $p_k\xleftarrow{\sigma_1} p_{k-1}\cdots p_{1}\xleftarrow{\sigma_k} p_0$.
Unless specified otherwise, an automaton is assumed to be a left automaton.

\end{rmk}

\begin{xmp}

The DFA of Figure~\ref{nfas}~(a) recognizes the language $L_{even}$ of words of even length.
The NFA of Figure~\ref{nfas}~(b) recognizes the language $L_{ends}$ of words over $\Sigma=\left\{ a,b \right\}$ that start and end with an $a$.

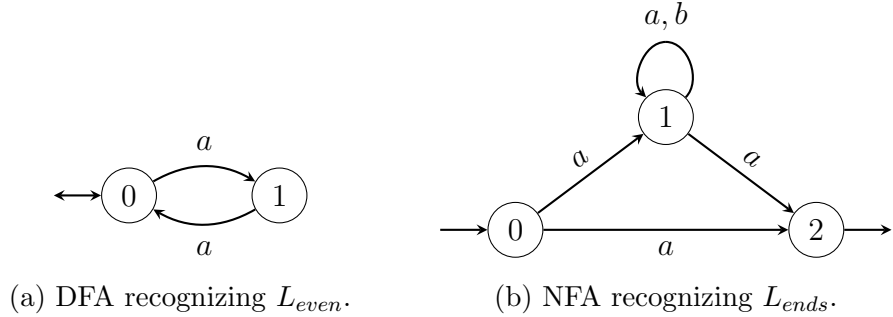
\begin{figure}[t]
        \centering
        \begin{subfigure}[b]{0.4\textwidth}
        	\centering
\begin{tikzpicture}
\node[blanc] (q0) at (0,0) {0} ;
\node[blanc] (q1) at (2,0) {1} ;
\draw[thick,->,>=stealth] (q0) to[bend left] (q1);
\draw[thick,->,>=stealth] (q1) to[bend left] (q0);

\draw[thick,<->,>=stealth] (-1,0) --(q0);

\node (n) at (1,0.7) {$a$};
\node (n) at (1,-0.7) {$a$};

\end{tikzpicture}
\caption{DFA recognizing $L_{even}$.}
        \end{subfigure}
        ~
        \begin{subfigure}[b]{0.4\textwidth}
        \centering
\begin{tikzpicture}
\node[blanc] (q0) at (0,0) {0} ;
\node[blanc] (q1) at (2,1.5) {1} ;
\node[blanc] (q3) at (4,0) {2} ;
\draw[thick,->,>=stealth] (-1,0) --(q0);
\draw[thick,<-,>=stealth] (5,0) --(q3);
\draw[thick,->,>=stealth] (q0)--(q1) node[midway,above,sloped] {$a$};
\draw[thick,->,>=stealth] (q0)--(q3) node[midway,below,sloped] {$a$};
\draw[thick,->,>=stealth] (q1)--(q3) node[midway,above,sloped] {$a$};
\draw[thick] (q1) to[out=45,in=0](2,2.5) node[above] {$a,b$};
\draw[thick,->,>=stealth] (2,2.5) to[out=180,in=135](q1);
\end{tikzpicture}
\caption{NFA recognizing $L_{ends}$.}
        \end{subfigure}
        \caption{Two automata.}
        \label{nfas}
\end{figure}

\end{xmp}

\section{Transducers}
A \emph{nondeterministic finite state transducer} (NFT) on an alphabet $\Sigma$ is a tuple $T=(Q,I,F,\Delta,i,t)$,
where $Q$ is a finite set of states,
$I\subseteq Q$ is the set of initial states,
$F\subseteq Q$ is the set of accepting states,
$\Delta: Q\times\Sigma\times Q\rightarrow \Sigma^\ast$ is a partial function defining the transitions \footnote{This type of transducer is sometimes called \emph{real-time}\cite{sakarovitch09}. In the general case, a transition of a transducer may be labelled by any pair of words},
$i:I\rightarrow \Sigma^\ast$ is the initial function.
and $t:F\rightarrow \Sigma^\ast$ is the terminal function.

Let $\Delta (p,\sigma,q)=v$ be a transition of $T$. The letter $\sigma$ and the word $v$ are respectively called the \emph{input} and the \emph{output} of the transition.
We denote $\Delta (p,\sigma,q)=v$ by $p\xrightarrow{\sigma\mid v}_T q$ (we will omit the $T$ when the context is clear).
We extend the relation to words: $p_0\xrightarrow{\sigma_1\ldots\sigma_k\mid v_1\ldots v_k} p_k$ when we have $p_0\xrightarrow{\sigma_1\mid v_1} p_1\xrightarrow{\sigma_2\mid v_2} p_2\cdots p_{k-1}\xrightarrow{\sigma_k\mid v_k} p_k$. Such a sequence is called a \emph{run of $T$ on $\sigma_1\ldots\sigma_k$}, and if $p_0\in I$ and $p_k\in F$, it is called a \emph{successful run}. We also add $q\xrightarrow {\varepsilon\mid \varepsilon} q$, $\forall q\in Q$.
An NFT realizes a transduction noted $\llbracket T \rrbracket\subseteq \Sigma^\ast \times \Sigma^\ast$. We also say that $T$ \emph{defines} the transduction $\llbracket T \rrbracket$.
The pair $(u,v)\in \llbracket T \rrbracket$ if there exist $q_0\xrightarrow{u\mid w}q_f$ such that $q_0\in I$, $q_f\in F$ and $v=i(q_o)\ w\ t(q_f)$.
An NFT is called \emph{functional} if it defines a functional transduction.
From now on the NFTs considered will be functional.

A \emph{deterministic finite state transducer} (DFT) on an alphabet $\Sigma$ is defined similarly to an NFT as a tuple $T=(Q,q_0,F,\delta,i,t)$ 
where the transition function is $\delta: Q\times\Sigma\rightarrow Q\times \Sigma^\ast$ and $q_0\in Q$.

The \emph{underlying automaton} of a transducer is the automaton obtained by removing the outputs from the transitions.
Let $T=(Q,I,F,\Delta,i,t)$ be a transducer, and let $\bar{\Delta}\subseteq Q\times \Sigma\times Q$ be the projection of $\Delta$. Then $A_T=(Q,I,F,\bar\Delta)$ is called the underlying automaton of $T$.

Functions defined by NFTs are called \emph{rational}, functions defined by DFTs are called \emph{subsequential} and functions defined by DFTs without terminal output are called \emph{sequential}.

\begin{xmp}
We give two examples of transducers:
The transducer of Figure~\ref{nfts}~(a) defines the function $f_{even}$ while the transducer of Figure~\ref{nfts}~(b) defines the function $f_{ends}$.

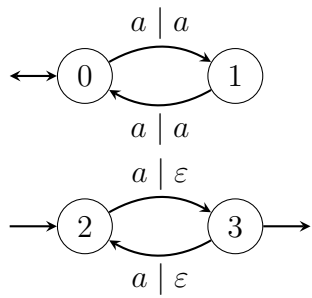
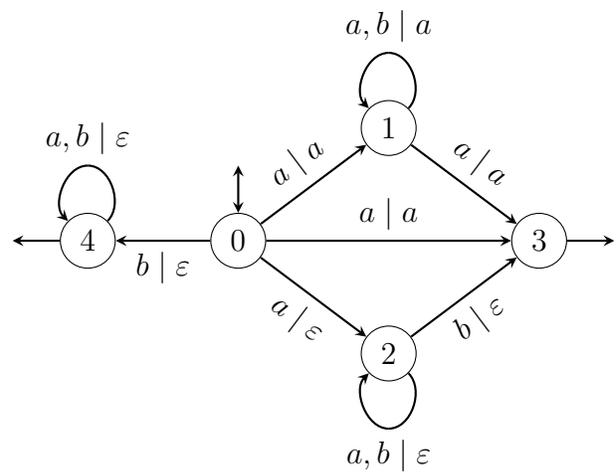
\begin{figure}[t]
        \centering
        \begin{subfigure}[b]{0.4\textwidth}
        	\centering
\begin{tikzpicture}
\node[blanc] (q0) at (0,0) {0} ;
\node[blanc] (q1) at (2,0) {1} ;
\draw[thick,->,>=stealth] (q0) to[bend left] (q1);
\draw[thick,->,>=stealth] (q1) to[bend left] (q0);

\draw[thick,<->,>=stealth] (-1,0) --(q0);

\node (n) at (1,0.7) {$a\mid a$};
\node (n) at (1,-0.7) {$a\mid a$};

\node[blanc] (q2) at (0,-2) {2} ;
\node[blanc] (q3) at (2,-2) {3} ;
\draw[thick,->,>=stealth] (q2) to[bend left] (q3);
\draw[thick,->,>=stealth] (q3) to[bend left] (q2);

\draw[thick,->,>=stealth] (-1,-2) --(q2);
\draw[thick,<-,>=stealth] (3,-2) --(q3);

\node (n) at (1,-1.3) {$a\mid \varepsilon$};
\node (n) at (1,-2.7) {$a\mid \varepsilon$};

\end{tikzpicture}
\caption{NFT defining $f_{even}$.}
        \end{subfigure}
        ~
        \begin{subfigure}[b]{0.4\textwidth}
        \centering
\begin{tikzpicture}
\node[blanc] (q0) at (0,0) {0} ;
\node[blanc] (q1) at (2,1.5) {1} ;
\node[blanc] (q2) at (2,-1.5) {2} ;
\node[blanc] (q3) at (4,0) {3} ;
\node[blanc] (q4) at (-2,0) {4} ;
\draw[thick,<->,>=stealth] (0,1) --(q0);
\draw[thick,<-,>=stealth] (5,0) --(q3);
\draw[thick,<-,>=stealth] (-3,0) --(q4);
\draw[thick,->,>=stealth] (q0)--(q1) node[midway,above,sloped] {$a\mid a$};
\draw[thick,->,>=stealth] (q0)--(q3) node[midway,above,sloped] {$a \mid a$};
\draw[thick,->,>=stealth] (q0)--(q2) node[midway,below,sloped] {$a\mid \varepsilon$};
\draw[thick,->,>=stealth] (q1)--(q3) node[midway,above,sloped] {$a\mid a$};
\draw[thick,->,>=stealth] (q2)--(q3) node[midway,below,sloped] {$b\mid \varepsilon$};
\draw[thick] (q1) to[out=45,in=0](2,2.5) node[above] {$a,b\mid a$};
\draw[thick,->,>=stealth] (2,2.5) to[out=180,in=135](q1);
\draw[thick] (q2) to[out=-45,in=0](2,-2.5) node[below] {$a,b\mid\varepsilon$};
\draw[thick,->,>=stealth] (2,-2.5) to[out=180,in=-135](q2);
\draw[thick,->,>=stealth] (q0)--(q4) node[midway,below,sloped] {$b\mid \varepsilon$};
\draw[thick] (q4) to[out=45,in=0](-2,1) node[above] {$a,b\mid \varepsilon$};
\draw[thick,->,>=stealth] (-2,1) to[out=180,in=135](q4);

\end{tikzpicture}
\caption{NFT defining $f_{ends}$.}
        \end{subfigure}
        \caption{Two transducers.}
        \label{nfts}
\end{figure}

\end{xmp}

\chapter{Logical and algebraic characterizations of classes of regular languages}

\section{Congruences and monoids}

Let $\Sigma$ be a finite alphabet and let $\sim$ be an equivalence relation on $\Sigma^\ast$.
We say that $\sim$ is a \emph{right congruence} if it verifies:
$$\forall u,v\in \Sigma^\ast, \forall \sigma\in \Sigma, u\sim v\Rightarrow u\sigma\sim v\sigma $$
Symmetrically we define a \emph{left congruence}.
A \emph{congruence} is defined as both a left and a right congruence.

Let $\sim_1$ and $\sim_2$ be two equivalence relations.
We say that $\sim_1$ is \emph{finer} than $\sim_2$, or that $\sim_2$ is \emph{coarser} than $\sim_1$, denoted by ${\sim_1}\sqsubseteq{\sim_2}$ if for all $u,v$ we have $u\sim_1v\Rightarrow u\sim_2v$.

We recall that a monoid $M$ is a set given with an associative binary operation (written multiplicatively) and an \emph{identity element} which is neutral (both left and right) for this operation.
A \emph{monoid homomorphism} (morphism for short) is an application $\mu: M_1\rightarrow M_2$ between two monoids  verifying $\mu(xy)=\mu(x)\mu(y)$ for $x,y\in M_1$ and $\mu(1_{M_1})=1_{M_2}$ with $1_{M_i}$ being the identity element of $M_i$ for $i\in \left\{1,2 \right\}$.
We say that a language $L$ of $\Sigma^\ast$ is \emph{recognized} by $M$ if there exist $P\subseteq M$ and $\mu: \Sigma^\ast\rightarrow M$ a monoid morphism such that $L=\mu^{-1}(P)$.
Let $\sim$ be a congruence on $\Sigma^\ast$.
We say that $L$ is \emph{recognized} by $\sim$ if there exists $P\subseteq \Sigma^\ast/\sim$ such that $L=\left\{u\mid [u]\in P \right\}$.

Let $M$ be a monoid and let $\mu:\Sigma^\ast \rightarrow M$.
Then $\sim_\mu$ defined by $u\sim_\mu v \Leftrightarrow \mu(u)=\mu(v)$ is a congruence, and any language recognized by $M$ is recognized by $\sim_\mu$ for some $\mu$.
Conversely, for a congruence $\sim$ we have a monoid $\Sigma^\ast/\sim$ with a binary operation induced by the natural morphism $u\mapsto [u]$ such that any language recognized by $\sim$ is recognized by $\Sigma^\ast/\sim$.
Hence there is a natural equivalence between recognizability by congruence and by monoid.

\begin{xmp}
Let $\Sigma=\left\{a\right\}$.
Let $\mu:\begin{array}{rcl}
\Sigma^\ast &\rightarrow &\left\{0,1\right\}\\
a^n &\mapsto & n\ \mathrm{mod}\ 2
\end{array}$ be the morphism from $\Sigma^\ast$ into the group of two elements (written additively).
Then $L_{even}=\mu^{-1}(0)$. The language $L_{even}$ is recognized by the group of two elements.

The corresponding congruence is defined by $a^m\sim a^n$ iff $m=n\ \mathrm{mod}\ 2$.
Then $L_{even}=[\varepsilon]$. The language $L_{even}$ is recognized by the congruence $\sim$.
\end{xmp}

\subsection{Syntactic congruences and monoids}
Let $L$ be a language. We define the \emph{syntactic congruence} of $L$ as:
$$u\approx_L v \Leftrightarrow \left( \forall x,y\quad xuy\in L \Leftrightarrow xvy\in L \right)$$
From this we can naturally define the \emph{syntactic monoid} of $L$,  $M(L)=\Sigma^\ast/\approx_L$.
According to the Myhill-Nerode Theorem \cite{nerode58}, a language is regular if and only if its syntactic congruence has finite index.

\begin{rmk}
\label{rmk2}
Let us also show that $\approx_L$ is the coarsest congruence recognizing $L$.
\end{rmk}

\begin{prf}
We see that, by definition, $L$ is a union of classes modulo $\approx_L$ hence the syntactic congruence of $L$ does recognize $L$.
Let $\sim$ be a congruence  recognizing $L$.
Since $\sim$ is a congruence, if $u\sim v$ then for all $x,y\in \Sigma^\ast$, $xuy\sim xvy$. This means that if $u\sim v$, for all $x,y\in \Sigma^\ast$, $xuy\in L \Leftrightarrow xvy\in L$. Hence $u\approx_L v$ and $\sim$ is finer than $\approx_L$.
\end{prf}

\begin{xmp}
Figure~\ref{mon1}~(a) represents the multiplication table of the syntactic monoid of $L_{ends}$, with elements labelled as representatives of congruence classes (omitting the identity element).

Figure~\ref{mon1}~(b) shows the automaton associated with the syntactic congruence of $L_{ends}$, with the states labelled as representatives of congruence classes.

The automaton associated with the syntactic congruence of $L_{even}$ is simply the automaton of Figure~\ref{nfas}.
The syntactic monoid of $L_{even}$ is the group with two elements.
\begin{figure}[t]
        \centering
        \begin{subfigure}[b]{0.4\textwidth}
\centering
\begin{tabular}{c||cccc}
&$a$&$b$&$ab$&$ba$\\
\hline
\hline
$a$&$a$&$ab$&$ab$&$a$\\
\hline
$b$&$ba$&$b$&$b$&$ba$\\
\hline
$ab$&$a$&$ab$&$ab$&$a$\\
\hline
$ba$&$ba$&$b$&$b$&$ba$\\

\end{tabular}
\caption{Multiplication table of the syntactic monoid of $L_{ends}$.}
        \end{subfigure}
        ~
        \begin{subfigure}[b]{0.4\textwidth}
\centering
\begin{tikzpicture}
\node[blanc,minimum height=9mm] (e) at (0,0) {$\varepsilon$} ;
\node[blanc,minimum height=9mm] (a) at (2,1.5) {$a$} ;
\node[blanc,minimum height=9mm] (ab) at (4,1.5) {$ab$} ;
\node[blanc,minimum height=9mm] (b) at (2,-1.5) {$b$} ;
\node[blanc,minimum height=9mm] (ba) at (4,-1.5) {$ba$} ;
\draw[thick,->,>=stealth] (-1,0) --(e);
\draw[thick,<-,>=stealth] (2,0.5) --(a);
\node[above] at (3,2) {$b$};
\node[below] at (3,1) {$a$};
\node[above] at (3,-1) {$a$};
\node[below] at (3,-2) {$b$};

\draw[thick,->,>=stealth] (e)--(a) node[midway,above,sloped] {$a$};
\draw[thick,->,>=stealth] (e)--(b) node[midway,below,sloped] {$b$};
\draw[thick,->,>=stealth] (a) to[bend left] (ab);
\draw[thick,->,>=stealth] (b) to[bend left] (ba);
\draw[thick,->,>=stealth] (ab) to [bend left] (a);
\draw[thick,->,>=stealth] (ba) to[bend left] (b);
\draw[thick] (a) to[out=45,in=0](2,2.5) node[above] {$a$};
\draw[thick,->,>=stealth] (2,2.5) to[out=180,in=135] (a);
\draw[thick] (ab) to[out=45,in=90](5,1.5) node[right] {$b$};
\draw[thick,->,>=stealth] (5,1.5) to[out=-90,in=-45](ab);

\draw[thick] (b) to[out=-45,in=0](2,-2.5) node[below] {$b$};
\draw[thick,->,>=stealth] (2,-2.5) to[out=180,in=-135] (b);
\draw[thick] (ba) to[out=45,in=90](5,-1.5) node[right] {$a$};
\draw[thick,->,>=stealth] (5,-1.5) to[out=-90,in=-45](ba);

\end{tikzpicture}
\caption{Automaton associated with the syntactic congruence of $L_{ends}$.}
        \end{subfigure}
        \caption{}
        \label{mon1}
\end{figure}

\end{xmp}

\subsection{Transition congruences and monoids}
Let $A$ be an automaton, then we define the \emph{transition congruence} of $A$ as:
$$u\equiv_A v \Leftrightarrow \left( \forall p,q\in Q \quad u\in L_{p,q} \Leftrightarrow v\in L_{p,q} \right)$$
From this we can naturally define the \emph{transition monoid} of $A$,  $M(A)=\Sigma^\ast/\equiv_A$.

\begin{rmk}
Let $A$ be an automaton. Then $L(A)$ is recognizable by $\equiv_A$.
Indeed, if $u\equiv_A v$ and $u \in L(A)$ then $v\in L(A)$.
Therefore $L(A)$ is the union of congruence classes.
This also implies that $L(A)$ is recognizable by the monoid $M(A)$.
\end{rmk}

Let $A$ be a left (resp. right) deterministic automaton, then we define the \emph{right} (rep. \emph{left}) \emph{congruence associated with $A$} as:
$$u\sim_A v \Leftrightarrow \left( \forall p\in Q \quad u\in L_{q_0,p} \Leftrightarrow v\in L_{q_0,p} \right)$$
Conversely, let $\sim$ be a right (resp. left) congruence.
We can define a unique deterministic left (resp. right) automaton without final state associated with $\sim$ as:
$A_\sim=(\Sigma^\ast/\sim,[\varepsilon],\delta)$ with $\delta$ being the natural right (resp. left) action of $\Sigma$ on $\Sigma^\ast/\sim$.

\begin{rmk}
\label{rmk1}
Let $A$ be a left (resp. right) automaton. If we assume that $A$ is trim, \emph{i.e.} any state appears in at least one successful run of $A$, then $u\equiv_A v$ if and only if for all $w\in \Sigma^\ast$, $wu\sim_A wv$ (resp. $uw\sim_A vw$).

Let us also remark that in that case, $\equiv_A$ is the coarsest congruence to be finer than $\sim_A$.
\end{rmk}
For a transducer $T$, we will use the notations $M(T)$, and $\equiv_T$ when referring to the transition monoid of the underlying automaton of $T$.

\section{Monoid varieties}
Let $U$ be a finite set of variables and $s_1,s_2\in U^\ast$.
We say that a monoid $M$ satisfies the equality $s_1=s_2$ if for any morphism $\eta:U^\ast\rightarrow M$, we have $\eta(s_1)=\eta(s_2)$.
We say that a congruence $\sim$ satisfies the equality $s_1=s_2$ if for any morphism $\eta:U^\ast\rightarrow\Sigma^\ast$, we have $\eta(s_1)\sim \eta(s_2)$.

A \emph{monoid pseudovariety} (variety for short) is defined \cite{straubing96} as a set of finite monoids that is closed under finite direct product, the taking of submonoids (\emph{i.e.} a subset stable by multiplication containing the same identity element) and of homomorphic images.
We will use the Theorem given in \cite{eilenbergs76} to characterize monoid varieties.
Let $E$ be a countable set of equalities.
Then we say that the set of finite monoids satisfying all but finitely many equalities in $E$ is the monoid variety \emph{ultimately defined} by $E$.
The Theorem states that any monoid variety is ultimately definable.
In the following we will use this characterization as a definition.

Let \textbf V be a monoid variety, we say that an automaton (resp. a transducer) is a \emph{\textbf V-automaton} (resp. \emph{\textbf V-transducer}) if its transition monoid is in \textbf V.
A language is called a \emph{\textbf V-language} if there exists a \textbf V-automaton recognizing it.
A congruence is called a \emph{\textbf V-congruence} if the monoid associated with it is in \textbf V.

\begin{prp}
\label{v-congruence}
Let \textbf V be a monoid variety.
\begin{itemize}
\item Let $\sim$ be a congruence. Then $\sim$ is a \textbf V-congruence if and only if it satisfies all but finitely many equalities defining \textbf V.

\item Let $\sim_1$, $\sim_2$ be two congruences such that $\sim_1\sqsubseteq \sim_2$. If $\sim_1$ is a \textbf V-congruence, then so is $\sim_2$.

\end{itemize}
\end{prp}

\begin{prf}
Let $U$ be a finite set of variables.
Let $M=\Sigma^\ast/\sim$ and let $\mu: \Sigma^\ast \rightarrow M$ be the natural morphism.
Let us assume that $\sim$ is a \textbf V-congruence.
Let $\eta:U^\ast\rightarrow \Sigma^\ast$ be a morphism.
Since $M\in \mathbf V$, we have by definition that if $M$ satisfies $s_1=s_2$ then $\mu\circ\eta(s_1)=\mu\circ\eta(s_2)$ which means that $\eta(s_1)\sim\eta(s_2)$.
Conversely, let us assume that $\sim$ satisfies all but finitely many equalities of \textbf V.
Let $\eta:U^\ast\rightarrow M$ be a morphism.
Let us show there exists a morphism $\nu:U^\ast \rightarrow \Sigma^\ast$ such that $\eta=\mu\circ\nu$.
We only have to give the image of any $u\in U$ to describe a morphism.
We define $\nu(u)=w_u$ with $w_u\in \mu^{-1}(\eta(u))$ (which is not empty since $\mu$ is surjective).
This yields a morphism which verifies $\eta=\mu\circ\nu$.
Hence we have that if $\sim$ satisfies  $s_1=s_2$ then $\nu(s_1)\sim\nu(s_2)$ which means that $\eta(s_1)=\eta(s_2)$.

Let $s_1,s_2\in U^\ast$ such that $\sim_1$ satisfies $s_1=s_2$,
\emph{i.e.} for any morphism $\eta: U^\ast \rightarrow \Sigma^\ast$, $\eta(s_1)\sim_1\eta(s_2)$.
In particular, $\eta(s_1)\sim_2\eta(s_2)$, which means that $\sim_2$ satisfies $s_1=s_2$.
\end{prf}

\begin{rmk}
Some of the most commonly used varieties are shown in Figure~\ref{table}.
A particular case is the variety of aperiodic monoids, denoted by \textbf A, which has been shown to recognize the languages definable in first-order logic \cite{schutzenberger65,mcnaughtonp71}.

For a given variety \textbf V, the \emph{membership problem}, \emph{i.e.} deciding if some monoid belongs to \textbf V, is not decidable \textit{a priori}.
For a discussion on this problem see \cite{pin97}.
However, if the membership problem is decidable, then the problem of deciding if a language is a \textbf V-language is also decidable. Indeed, according to Remark~\ref{rmk2} and Proposition~\ref{v-congruence}, one only has to check that the syntactic monoid of the language belongs to \textbf V.

\end{rmk}

\section{Logics}

Given an alphabet $\Sigma$, \emph{monadic second-order formulas} (MSO formulas) are built over first order variables $x,y,\ldots$ and second order variables $X,Y,\ldots$. Atomic formulas are built using the binary predicate $<$ and for each $\sigma \in \Sigma$ a unary predicate, also denoted by $\sigma$. The formulas are defined by the following grammar:

$$\varphi \mathrel{::=}\exists X\ \varphi\mid \exists x\ \varphi\mid \varphi\wedge \varphi \mid \neg \varphi\mid x\in X \mid \sigma(x) \mid x<y\mid (\varphi)$$
Universal quantifiers and other boolean connectives are defined as usual:
$\forall X\ \varphi \equiv \neg \exists X\ \neg \varphi$, $\forall x\ \varphi \equiv \neg \exists x\ \neg \varphi$, $\varphi_1\vee\varphi_2\equiv\neg (\neg \varphi_1 \wedge \neg \varphi_2)$,
$\varphi_1\rightarrow\varphi_2\equiv\neg \varphi_1 \vee \varphi_2$.
We also define the constant formulas $\top$ and $\bot$ which are respectively always and never satisfied.

We do not define the semantics of MSO formulas, neither the standard notion of free and bound variables, but rather give examples and refer the reader to \cite{ebbinghausF95} for formal definitions.

Given a closed formula $\varphi$, and a word $w$ satisfying $\varphi$, we write $w\models \varphi$ and $L(\varphi)$ denotes the language of words satisfying $\varphi$.

A \emph{logical fragment} of MSO (logic for short) $\mathcal F$ is a subset of formulas of MSO. Given a language $L$, we say that $L$ is $\mathcal F$-definable (or an $\mathcal F$-language) is there exists a closed formula $\varphi\in \mathcal F$ such that $L=L(\varphi)$.

\begin{xmp}
We define the successor relation as:
$$\mathsf{succ}(x,y)=(x<y)\wedge (\neg\exists z\ (x<z)\wedge(z<y))$$
We define the minimum and maximum unary predicates:
$$\mathsf{min}(x)=\neg \exists y\ (y<x)$$
$$\mathsf{max}(x)=\neg \exists y\ (x<y)$$
Then we define the MSO formula $\varphi_{even}$ which defines the language of words of even length:
$$\begin{array}{rl}\varphi_{even}=\exists X\ \exists Y\ \forall x\ \forall y\ & x\in  X\leftrightarrow \neg(x\in Y)\\
\wedge&\mathsf{min}(x)\rightarrow x\in X\\
\wedge& \mathsf{max}(y)\rightarrow y\in Y\\
\wedge& \mathsf{succ}(x,y)\rightarrow (x\in X\wedge y\in Y)\vee (x\in Y\wedge y\in X)
\end{array}$$
\end{xmp}
We give below several logics and refer to Figure~\ref{table} to see the languages defined by these logics.

\emph{First-order} formulas (FO) denotes the logic of formulas without any second order variables.
\begin{xmp}
Let $\Sigma=\left\{a,b\right\}$.
We define the formula $\varphi_{ends}$ which defines the language $L_{ends}$:
$$\varphi_{ends}=\exists x\ \exists y\ \mathsf{min}(x)\wedge a(x) \wedge\mathsf{max}(y)\wedge a(y)$$
\end{xmp}

\emph{MSO$[=]$} denotes the logic of formulas without the binary predicate $<$.

Let $k$ be a non-negative integer, then \emph{FO$^k$} denotes the logic of FO formulas with only $k$ variables.

Let \textbf V be a monoid variety and let $\mathcal F$ be a logic.
We say that $\mathcal F$ is a logic \emph{corresponding to} \textbf V if for any language $L$ we have:
$L$ is an $\mathcal F$-language if and only if $L$ is a \textbf V-language.

\chapter{Logic-bimachine correspondence}

\section{Bimachines}

Here we will describe the notion of a bimachine, introduced by Sch\"utzenberger in \cite{schutzenberger61}, which is in some sense both left and right sequential (or subsequential).
Bimachines have been shown, see \emph{e.g.} \cite{berstelb79}, to define exactly rational functions and are in that sense equivalent to functional NFTs.

A bimachine is given by a tuple $B=(L,R,\omega,\lambda,\rho)$ where:
\begin{itemize}
\item $L$ is the set of states of a deterministic left automaton given with an initial state $l_0$.
\item $R$ is the set of states of a deterministic right automaton given with an initial state $r_0$.
\item The output function $\omega:L\times \Sigma \times R\rightarrow \Sigma^\ast$ which may be partial.
\item The terminal left function $\rho:L\rightarrow \Sigma^\ast$ and the terminal right function $\lambda:R\rightarrow \Sigma^\ast$ which may both be partial.
\end{itemize}

The output function is extended to words in this fashion:
Let $u=\sigma_1\ldots\sigma_n \in \Sigma^n$, let $l\xrightarrow{\sigma_1}l_1\cdots l_{n-1}\xrightarrow{\sigma_n} l_n$ be an execution of $L$ on $u$ from $l$ and let $r_n\xleftarrow{\sigma_1}r_{n-1}\cdots r_1\xleftarrow{\sigma_n} r$ be an execution of $R$ on $u$ from $r$. Then, if it is defined:
$$\omega(l,u,r)=\omega(l,\sigma_1,r_{n-1})\cdots\omega(l_{n-1},\sigma_n,r)$$
If $l=l_0$ and $r=r_0$, when it exists the image of $u$ by $B$ is defined as:
$$\llbracket B \rrbracket(u)=\lambda(r_n)\omega(l_0,u,r_0)\rho(l_n)$$.

\begin{rmk}
In order to simplify, we will write $L$ (resp. $R$) when referring both to the set of states and the automaton.
\end{rmk}
\begin{rmk}
When $R$ is reduced to a single element then the bimachine is simply a subsequential transducer (sequential if $\rho$ is a constant function equal to $\varepsilon$ on its domain).
\end{rmk}

\begin{xmp}
\label{xmp-bim}
We describe a bimachine $B=(L,R,\omega,\lambda,\rho)$ defining $f_{ends}$.
The automata $L=\left\{l_0,l_a,l_b\right\}$ and $R=\left\{r_0,r_a,r_b\right\}$ are identical (except that one is a left automaton and the other is a right automaton) and $L$ is shown in Figure~\ref{bim1}. The functions $\lambda$ and $\rho$ are constant and equal to $\varepsilon$,
$\omega(l,\sigma,r)=a$ if $l\in\left\{l_0,l_a\right\}$, $\sigma=a$ and $r\in\left\{r_0,r_a\right\}$ or if $l=l_a$, $\sigma=b$ and $r=r_a$. In the other cases $\omega(l,\sigma,r)=\varepsilon$.
In Figure~\ref{exe} we show the execution of this bimachine on the word $abaa$.

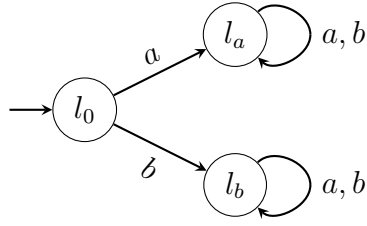
\begin{figure}[t]
\centering
\begin{tikzpicture}
\node[blanc] (l0) at (0,0) {$l_0$} ;
\node[blanc] (l1) at (2,1) {$l_a$} ;
\node[blanc] (l2) at (2,-1) {$l_b$} ;
\draw[thick,->,>=stealth] (-1,0) --(l0);
\draw[thick,->,>=stealth] (l0)--(l1) node[midway,above,sloped] {$a$};
\draw[thick,->,>=stealth] (l0)--(l2) node[midway,below,sloped] {$b$};

\draw[thick] (l1) to[out=45,in=90](3,1) node[right] {$a,b$};
\draw[thick,->,>=stealth] (3,1) to[out=-90,in=-45](l1);
\draw[thick] (l2) to[out=45,in=90](3,-1) node[right] {$a,b$};
\draw[thick,->,>=stealth] (3,-1) to[out=-90,in=-45](l2);

\end{tikzpicture}
\caption{Left automaton of $B$.}
\label{bim1}
\end{figure}

\begin{figure}[t]
\centering
\begin{tikzpicture}
\node[minimum size=35pt] (l0) at (0,0) {$l_0$} ;
\node[minimum size=35pt] (l1) at (2,0) {$l_a$} ;
\node[minimum size=35pt] (l2) at (4,0) {$l_a$} ;
\node[minimum size=35pt] (l3) at (6,0) {$l_a$} ;
\node[minimum size=35pt] (l4) at (8,0) {$l_a$} ;

\node (u1) at (1,-1) {$a$} ;
\node (u2) at (3,-1) {$b$} ;
\node (u3) at (5,-1) {$a$} ;
\node (u4) at (7,-1) {$a$} ;

\node[minimum size=35pt] (r0) at (8,-2) {$r_0$} ;
\node[minimum size=35pt] (r1) at (6,-2) {$r_a$} ;
\node[minimum size=35pt] (r2) at (4,-2) {$r_a$} ;
\node[minimum size=35pt] (r3) at (2,-2) {$r_a$} ;
\node[minimum size=35pt] (r4) at (0,-2) {$r_a$} ;

\node (v0) at (0,-3) {$\varepsilon$} ;
\node (v1) at (2,-3) {$a$} ;
\node (v2) at (4,-3) {$a$} ;
\node (v3) at (6,-3) {$a$} ;
\node (v4) at (8,-3) {$a$} ;
\node (v5) at (10,-3) {$\varepsilon$} ;

\draw[thick,->,>=stealth] (l0) --(l1);
\draw[thick,->,>=stealth] (l1) --(l2);
\draw[thick,->,>=stealth] (l2) --(l3);
\draw[thick,->,>=stealth] (l3) --(l4);

\draw[thick,->,>=stealth] (r0) --(r1);
\draw[thick,->,>=stealth] (r1) --(r2);
\draw[thick,->,>=stealth] (r2) --(r3);
\draw[thick,->,>=stealth] (r3) --(r4);

\draw[rotate=45] (-0.3,0.3) rectangle (0.3,-3.3);
\draw[xshift=2cm,rotate=45] (-0.3,0.3) rectangle (0.3,-3.3);
\draw[xshift=4cm,rotate=45] (-0.3,0.3) rectangle (0.3,-3.3);
\draw[xshift=6cm,rotate=45] (-0.3,0.3) rectangle (0.3,-3.3);

\draw (11,1) -- (11,-4);

\node at (13,0) {Run of $L$};
\node at (13,-1) {Input};
\node at (13,-2) {Run of $R$};
\node at (13,-3) {Output};

\end{tikzpicture}
\caption{Execution of a bimachine}
\label{exe}
\end{figure}
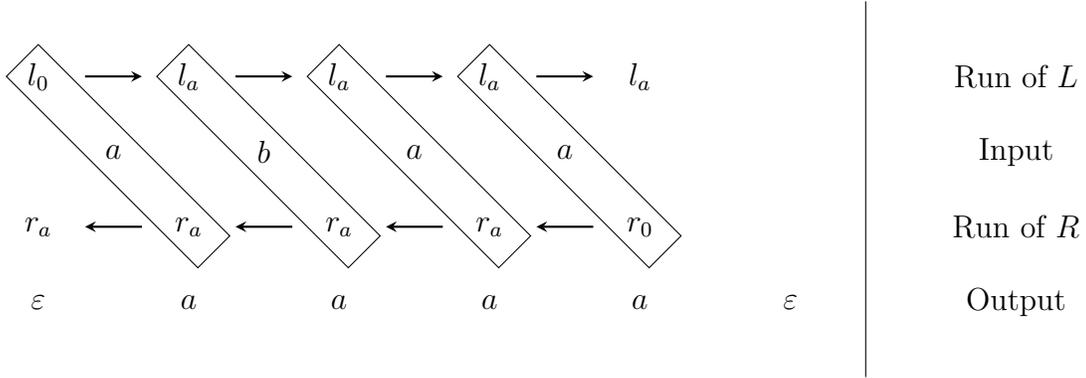
\end{xmp}
A bimachine is called \emph{complete} if its output function is total.
\begin{prp}
\label{complete}
Any bimachine is equivalent to some complete bimachine.
\end{prp}

\begin{prf}
The idea of the proof is to extend the output function by outputting $\varepsilon$ when it is not defined. To ensure that the domain of the function is preserved, we refine the left (for instance) automaton to only accept words in the domain.

Let $B=(L,R,\omega,\lambda,\rho)$ be a bimachine defining a function $f$. We construct $B'=(L',R,\omega',\lambda,\rho')$ a complete bimachine equivalent to $B$.
Let $A$ be an automaton recognizing $\mathrm{dom}(f)$, which is rational.
We define $L'=L\times A$.
Let $(l,q)$ be a state of $L'$ and let $r\in R$. Then $\omega'((l,q),\sigma,r)=\omega(l,\sigma,r)$ if it is defined and $\omega'((l,q),\sigma,r)=\varepsilon$ otherwise.
If $q$ is a final state of $A$ then $\rho'(l,q)=\rho(l)$, otherwise $\rho'(l,q)$ is not defined.

If a word belongs to $\mathrm{dom}(f)$ then the outputs in the run of $B'$ are the same as the ones in $B$. If a word does not belong to $\mathrm{dom}(f)$ then it is not accepted by $L'$.
Hence $B$ and $B'$ are equivalent.
\end{prf}

\subsection{From bimachines to transducers}
\label{totrans}
Let $B=(L,R,\omega,\lambda,\rho)$ be a bimachine defining $f$, we can construct an NFT $T=(Q,I,F,\Delta,i,t)$ defining the same function:

\begin{itemize}
\item $Q=L\times R$
\item $I=\left\{l_0 \right\}\times \mathrm{dom}(\lambda)$.
\item $F=\mathrm{dom}(\rho) \times \left\{r_0\right\}$.
\item $\Delta((l_1,r_1),\sigma,(l_2,r_2))=v$\\ if
$l_1\xrightarrow{\sigma}l_2$, $r_1\xleftarrow{\sigma}r_2$ and $\omega(l_1,\sigma,r_2)=v$.
\item $i:(l_0,r) \mapsto \lambda(r)$ whenever $(l_0,r)\in I$.
\item $t:(l,r_0) \mapsto \rho(l)$ whenever $(l,r_0)\in F$.

\end{itemize}

\subsection{From transducers to bimachines}
\label{tobim}

An NFA $A=\left(Q,I,F,\Delta \right)$ is called \emph{unambiguous} is for any word $u=\sigma_1\ldots \sigma_n\in \Sigma^n$ there exists at most one sequence $q_0\xrightarrow{\sigma_1}q_1\cdots q_{n-1}\xrightarrow{\sigma_n}q_n$ with $q_0\in I$ and $q_n\in F$.
An NFT is called \emph{unambiguous} if its underlying automaton is unambiguous.
It has been shown that any functional NFT is equivalent to some unambiguous NFT \cite{berstelb79}.

Let $T=(Q,I,F,\Delta,i,t)$ be an unambiguous NFT,
we can naturally describe an equivalent bimachine $B=(L,R,\omega,\lambda,\rho)$ with:

\begin{itemize}
\item $L=\Sigma^\ast / \equiv_T$ with the right action of $\Sigma$ as its transition function: $[u]\xrightarrow \sigma [u\sigma]$ and $l_0=[\varepsilon]$.
\item $R=\Sigma^\ast / \equiv_T$ with the left action of $\Sigma$ as its transition function: $[\sigma u]\xleftarrow \sigma [u]$ and $r_0=[\varepsilon]$.
\item $\omega(l,\sigma,r)$ is defined as the unique (by unambiguity) $v$ such that
$\exists u_l\in l, u_r\in r$ and $p,q\in Q$ with $u_l\in L_{I,p}$, $u_r\in L_{q,F}$ and $\Delta(p,\sigma,q)=v$.
\item $\Delta((l_1,r_1),\sigma,(l_2,r_2))=v$ if $l_1\xrightarrow{\sigma}l_2$, $r_1\xleftarrow{\sigma}r_2$ and $\omega(l_1,\sigma,r_2)=v$.
\item $\lambda: \begin{array}{c} R\rightarrow \Sigma^\ast \\ \left[ u\right] \mapsto i(p) \end{array}$ where $p$ is the unique (by unambiguity) $p\in I$ such that $u\in L_{p,F}$.
\item $\rho: \begin{array}{c} R\rightarrow \Sigma^\ast \\ \left[ u\right] \mapsto t(p) \end{array}$ where $p$ is the unique (by unambiguity) $p\in F$ such that $u\in L_{I,p}$.

\end{itemize}

\section{Logical transduction}

In the paper \cite{filiot15} an equivalence is given between NFTs and so called \emph{order-preserving} logical transducers.
The following definition coincides with this class of order-preserving logical transductions, in the logics of MSO$[<]$ and FO$[<]$.
It is also somewhat similar to the notion of logical translation given by \cite{mckenzieSTV06} in the particular case of length-preserving transductions.

Let $\mathcal F$ be a logic.
An $\mathcal F$-translation is given by a tuple
$$\mathcal T=\left(k,S,\left(\varphi_{j,\sigma,v}^<,\varphi_{j,\sigma,v}^>\right)_{\begin{smallmatrix} 1\leq j\leq k\\ \sigma \in \Sigma, v\in S \end{smallmatrix}},\left(\varphi_v^i\right)_{v\in S},\left(\varphi_v^t\right)_{v\in S}\right)$$
$\mathcal T$ is composed of $k>0$ and $S\subset \Sigma^\ast$ finite, the set of outputs. For $1\leq j\leq k$, $\sigma\in\Sigma$ and $v\in S$, $\varphi_{j,\sigma,v}^<$, $\varphi_{j,\sigma,v}^>$, $\varphi_v^i$ and $\varphi_v^t$ are closed $\mathcal F$-formulas.

By definition $(u,w)\in\llbracket \mathcal T \rrbracket$ if:
\begin{quote}
for $u=\sigma_1\ldots \sigma_n\in \Sigma^n$ there exists a decomposition $w=w_0w_1\ldots w_nw_{n+1}\in S^{n+2}$ such that:
\begin{itemize}
\item for $i \in \left\{1,\ldots,n \right\}$ there exists $j$, with $1\leq j\leq k$ such that we have both:
$\sigma_1\ldots \sigma_{i-1}\models \varphi_{j,\sigma_i,w_i}^<$ and 
$\sigma_{i+1}\ldots \sigma_n\models \varphi_{j,\sigma_i,w_i}^>$.
\item $u\models \varphi_{w_0}^i$ and $u\models \varphi_{w_{n+1}}^t$.
\end{itemize}

\end{quote}

We want to ensure that the domain of the translation is an $\mathcal F$-language.
For this we add the following condition of \emph{exhaustiveness}:
\begin{quote}
For any $\sigma\in \Sigma$, $u,w\in \Sigma^\ast$, there exists a pair $\left(\varphi_{j,\sigma,v}^<,\varphi_{j,\sigma,v}^>\right)$ with $j\leq k$ and $v\in S$ such that $u\models \varphi_{j,\sigma,v}^<$ and $v\models \varphi_{j,\sigma,v}^>$.

\end{quote}

\begin{rmk}
An $\mathcal F$-translation $\mathcal T=\left(k,S,\left(\varphi_{j,\sigma,v}^<,\varphi_{j,\sigma,v}^>\right)_{\begin{smallmatrix} 1\leq j\leq k\\ \sigma \in \Sigma, v\in S \end{smallmatrix}},\left(\varphi_v^i\right)_{v\in S},\left(\varphi_v^t\right)_{v\in S}\right)$ might not define a function but a relation.
In order to ensure that it is functional, we have to add some conditions:
\begin{itemize}
\item Let $j,j'\leq k,\sigma\in \Sigma, v_1\neq v_2\in S$. Either $\varphi_{j,\sigma,v_1}^< \wedge \varphi_{j',\sigma,v_2}^< $ is unsatisfiable or $\varphi_{j,\sigma,v_1}^> \wedge \varphi_{j',\sigma,v_2}^> $ is unsatisfiable.
\item Let $v_1\neq v_2\in S$, then $\varphi_{v_1}^i \wedge \varphi_{v_2}^i$ is unsatisfiable.
\item Let $v_1\neq v_2\in S$, then $\varphi_{v_1}^t \wedge \varphi_{v_2}^t$ is unsatisfiable.
\end{itemize}
In the following we will always assume that a translation satisfies these conditions.
\end{rmk}

\begin{xmp}
We give a translation defining $f_{ends}$.
Let $\Sigma=\left\{a,b\right\}$ and $\mathcal T=\left(k,S,\left(\varphi_{j,\sigma,v}^<,\varphi_{j,\sigma,v}^>\right)_{\begin{smallmatrix} 1\leq j\leq k\\ \sigma \in \Sigma, v\in S \end{smallmatrix}},\left(\varphi_v^i\right)_{v\in S},\left(\varphi_v^t\right)_{v\in S}\right)$ with $k=2$ and $S=\left\{\varepsilon,a\right\}$.
\begin{itemize}
\item $\varphi_{1,a,\varepsilon}^<=\exists x\ \mathsf{min}(x) \wedge b(x)$ and $\varphi_{1,a,\varepsilon}^>=\top$.
\item $\varphi_{2,a,\varepsilon}^<=\top$ and $\varphi_{2,a,\varepsilon}^>=\exists x\ \mathsf{max}(x) \wedge b(x)$.
\item $\varphi_{1,a,a}^<=\neg \exists x\ \mathsf{min}(x) \wedge b(x)$ and $\varphi_{1,a,a}^>=\neg \exists x\ \mathsf{max}(x) \wedge b(x)$
\item $\varphi_{1,b,\varepsilon}^<=\neg\exists x\ \mathsf{min}(x) \wedge a(x)$ and $\varphi_{1,b,\varepsilon}^>=\top$.
\item $\varphi_{2,b,\varepsilon}^<=\top$ and $\varphi_{2,b,\varepsilon}^>=\neg\exists x\ \mathsf{max}(x) \wedge a(x)$.
\item $\varphi_{1,b,a}^<=\exists x\ \mathsf{min}(x) \wedge a(x)$ and $\varphi_{1,b,a}^>=\exists x\ \mathsf{max}(x) \wedge a(x)$
\item $\varphi_\varepsilon^i=\top$ and $\varphi_\varepsilon^t=\top$.
\item All the other formulas are set to $\bot$.
\end{itemize}

\end{xmp}

\subsection{From machines to logics}

\begin{prp}
\label{prpbt}
Let \textbf V be a monoid variety and let $\mathcal F$ be a corresponding logic.
Let $f$ be a function definable by a complete \textbf V-bimachine. Then $f$ is definable by an $\mathcal F$-translation.
\end{prp}

\begin{prf}
Let $B=(L,R,\omega,\lambda,\rho)$ be a complete \textbf V-bimachine.
We can assume that the automata $L$ and $R$ are complete (for instance by taking the automata associated with $\sim_L$ and $\sim_R$, respectively).
Let $l\in L$, then there is an $\mathcal F$-formula $\varphi_l^L$ recognizing the language $L_{l_0,l}(L)$. Similarly for $r\in R$ let $\varphi_r^R$ be an $\mathcal F$-formula recognizing $L_{r_0,r}(R)$.
Let us define $\mathcal T=\left(k,S,\left(\varphi_{j,\sigma,v}^<,\varphi_{j,\sigma,v}^>\right)_{\begin{smallmatrix} 1\leq j\leq k\\ \sigma \in \Sigma, v\in S \end{smallmatrix}},\left(\varphi_v^i\right)_{v\in S},\left(\varphi_v^t\right)_{v\in S}\right)$ with $k=|L|\cdot|R|$, $S$ being the set of outputs in $B$.
For convenience, we will index the formulas with pairs of states instead of integer $\leq |L|\cdot|R|$.
Let $l\in L,\sigma\in \Sigma, r\in R,v\in S$ such that $\omega(l,\sigma,r)=v$.
We define $\varphi_{(l,r),\sigma,v}^<=\varphi_l^L$ and $\varphi_{(l,r),\sigma,v}^>=\varphi_r^R$.
We also define for $v\in S$, $\varphi_v^i=\bigwedge_{\lambda(r)=v}\varphi_r^R$ and $\varphi_v^t=\bigwedge_{\rho(l)=v}\varphi_l^L$.
We can see that since $L$ and $R$ are complete and $\omega$ is total, the translation is indeed exhaustive.

Let $(u,w)\in\llbracket B \rrbracket$.
Let $u=\sigma_1\ldots\sigma_n \in \Sigma^n$, let $l_0\xrightarrow{\sigma_1}l_1\cdots l_{n-1}\xrightarrow{\sigma_n} l_n$ be the execution of $L$ on $u$ and let $r_n\xleftarrow{\sigma_1}r_{n-1}\cdots r_1\xleftarrow{\sigma_n} r_0$ be the execution of $R$ on $u$.
Let $w=v_0v_1\ldots v_nv_{n+1}$ such that $\omega(l_{i-1},\sigma_i,r_{n-i})=v_i$ for $1\leq i\leq n$, $\lambda(r_n)=v_0$ and $\rho(l_n)=v_{n+1}$.
Then for $1\leq i\leq n$ we have both $\sigma_1\ldots\sigma_{i-1}\models \varphi_{(l_{i-1},r_{n-i}),\sigma_i,v_i}^<$ and $\sigma_{i+1}\ldots\sigma_n\models \varphi_{(l_{i-1},r_{n-i}),\sigma_i,v_i}^>$. We also have $u\models \varphi_{v_0}^i$ and $u\models \varphi_{v_{n+1}}^t$.
Hence $(u,w)\in\llbracket \mathcal T \rrbracket$.
Conversely let us assume $(u,w)\in\llbracket \mathcal T \rrbracket$.
With the same notations as above, we have for $1\leq i\leq n$, $\sigma_1\ldots\sigma_{i-1}\models \varphi_{(l_{i-1},r_{n-i}),\sigma_i,v_i}^<$ and $\sigma_{i+1}\ldots\sigma_n\models \varphi_{(l_{i-1},r_{n-i}),\sigma_i,v_i}^>$. We also have $u\models \varphi_{v_0}^i$ and $u\models \varphi_{v_{n+1}}^t$.
We just have to remark that the sequence of left (resp. right) states corresponds to an execution on $L$ (resp. $R$).
Then we have automatically that $v_i=\omega(l_{i-1},\sigma_i,r_{n-i})$, $\lambda(r_n)=v_0$ and $\rho(l_n)=v_{n+1}$. Thus we obtain $(u,w)\in\llbracket B \rrbracket$.
\end{prf}

\subsection{From logics to machines}

\begin{prp}
\label{t-to-b}
Let \textbf V be a monoid variety and let $\mathcal F$ be a corresponding logic.
Let $f$ be a function definable by an $\mathcal F$-translation. Then $f$ is definable by a complete \textbf V-bimachine.
\end{prp}

\begin{prf}
Let $\mathcal T=\left(k,S,\left(\varphi_{j,\sigma,v}^<,\varphi_{j,\sigma,v}^>\right)_{\begin{smallmatrix} 1\leq j\leq k\\ \sigma \in \Sigma, v\in S \end{smallmatrix}},\left(\varphi_v^i\right)_{v\in S},\left(\varphi_v^t\right)_{v\in S}\right)$ be an $\mathcal F$-translation.

We will describe a bimachine $B=(L,R,\omega,\lambda,\rho)$ equivalent to $\mathcal T$.
We define two sets of formulas: $F_L=\left\{\varphi_{j,\sigma,v}^<\mid 1\leq j\leq k,\ \sigma\in\Sigma,\ v\in S\right\}\cup \left\{\varphi_v^t \mid v\in S\right\}$ and $F_R=\left\{\varphi_{j,\sigma,v}^>\mid 1\leq j\leq k,\ \sigma\in\Sigma,\ v\in S\right\}\cup \left\{\varphi_v^i \mid v\in S\right\}$.
For any formula in $F_L$ there is a \textbf V-congruence recognizing the same language. Let $\sim_L$ be a \textbf V-congruence finer than all of these.
Similarly we can define $\sim_R$ and we obtain naturally the automata $L$, $R$.
Then for $[u]\in L$, $[w]\in R$ and $\sigma\in \Sigma$, $\omega([u],\sigma,[w])=v$ with $v$ being the unique word in $S$ such that for some $j\leq k$, $u\models \varphi_{j,\sigma,v}^<$ and $w\models \varphi_{j,\sigma,v}^>$.
Let us remark that $\omega$ is total by exhaustiveness.
Hence we obtain a complete \textbf V-bimachine and exactly as above, we can show it defines $f$.
\end{prf}

\begin{rmk}
If $\mathcal F$ is a logic corresponding to a variety \textbf V, it can easily be seen that $\mathcal F$-translations with the formulas $\varphi_{j,\sigma,v}^>=\top$ are equivalent to \textbf V-DFTs. Moreover, if we set for $v\neq \varepsilon$ that $\varphi_{v}^t=\bot$, we obtain the functions definable by \textbf V-DFTs with no terminal output.
\end{rmk}

\chapter{Algebraic characterization of classes of subsequential functions}
\label{chap4}

The first part of this chapter deals with a minimization procedure described by \cite{choffrut03} and we show it preserves the equalities of the transition congruence.
In particular, one can decide if a subsequential function is definable by an aperiodic DFT.
In the second part we show that a determinization algorithm of transducers preserves the aperiodicity of the transition monoid.

\section{Minimization}
\label{mindft}
In order to have an algebraic classification of functions, we need a notion similar to the one of minimal automaton. In the case of subsequential functions, there is such notion which we call the minimal DFT.
We describe the construction of a minimal DFT given in \cite{choffrut03} and prove that it preserves equality in the transition congruence.
However we do not prove here that it is correct and refer the reader to the original paper.
This minimization is canonical and not machine dependent.

Let $f$ be a subsequential function definable by some DFT $T =(Q,q_0,F,\delta,i,t)$.
For $u\in \Sigma^\ast$, we define when it exists $\widehat f(u)=\wedge\left\{f(uw)\mid  uw\in \mathrm{dom}(f) \right\}$.
We define the \emph{syntactic congruence} of $f$, for $u,v\in \Sigma^\ast$, $u\sim_f v$ if:

For all $w\in \Sigma^\ast$, $uw\in \mathrm{dom}(f) \Leftrightarrow vw\in \mathrm{dom}(f)$ and when it is defined, $\widehat f(u)^{-1} f(uw) = \widehat f(v)^{-1}f(vw)$.\\
Then the transducer $T_f=(Q_f,q_{0f},F_f,\delta_f,i_f,t_f)$ is equivalent to $T$ with:
\begin{itemize}
\item $Q_f=\left\{ [u]\mid  u\in \mathrm{dom}(\widehat f)\right\}$
\item $q_{0f}=[\varepsilon]$ (we assume that the domain of $T$ is not empty).
\item $\delta([u],\sigma)=([u\sigma],\widehat f(u)^{-1}\widehat f(u\sigma))$ when $\widehat f(u\sigma)$ is defined.
\item $F_f=\left\{[u]\mid  u\in \mathrm{dom}(f) \right\}$.
\item $i_f([\varepsilon])=\widehat f(\varepsilon)$.
\item $t_f([u])=\widehat f(u)^{-1}f(u)$ whenever $[u]\in F_f$.
\end{itemize}

\begin{prp}
\label{prp4}
Given a monoid variety \textbf V, let $f$ be a subsequential function. The function $f$ is definable by a \textbf V-DFT if and only if its minimal transducer is a \textbf V-DFT.
\end{prp}
\begin{prf}
According to Proposition~\ref{v-congruence} we only have to show ${\equiv_T}\sqsubseteq{\equiv_{T_f}}$.
Let $f$ be a subsequential function definable by some DFT $T =(Q,q_0,F,\delta,i,t)$. For $q\in Q$ let $T_q=(Q,q,F,\delta,\varepsilon,t)$ and let $f_q$ be the corresponding function.
Let $m_1,m_2\in \Sigma^\ast$ two words such that $m_1\equiv _{T}m_2$.
We construct $T_f$ using the minimization in \cite{choffrut03}. Let $v$ be a word in $\Sigma^\ast$.
Let  $q_0\xrightarrow{vm_1\mid x_1}p$ (if $p$ does not exist, we can add a garbage state in $T$).
Then obviously $q_0\xrightarrow{vm_2\mid x_2}p$ and for all $w\in \Sigma^\ast$, $vm_1w\in \mathrm{dom}(f) \Leftrightarrow vm_2w\in \mathrm{dom}(f)$.
Then when it is defined:
$$\widehat f(vm_1)=ix_1\bigwedge \left\{ f_p(w)\mid  w\in \mathrm{dom} (f_p) \right\}$$
$$\widehat f(vm_2)=ix_2\bigwedge \left\{ f_p(w)\mid  w\in \mathrm{dom} (f_p) \right\}$$
Let $s=\bigwedge \left\{ f_p(w)\mid  w\in \mathrm{dom} (f_p) \right\}$,
then for any $w\in \Sigma^\ast$ such that $vm_1w\in \mathrm{dom}(f)$:
$$\widehat f(vm_1)^{-1}f(vm_1w)=s^{-1}x_1^{-1}i^{-1}ix_1f_p(w)=s^{-1}f_p(w)$$
$$\widehat f(vm_2)^{-1}f(vm_2w)=s^{-1}x_2^{-1}i^{-1}ix_2f_p(w)=s^{-1}f_p(w)$$
Hence $vm_1\sim_f vm_2$ for any $v\in \Sigma^\ast$.
We have in $T_f$ for any $v\in \Sigma^\ast$:
$[v]\xrightarrow{m_1}[vm_1]=[vm_2]$ and $[v]\xrightarrow{m_2}[vm_2]$.
Hence $m_1\equiv _{T_f}m_2$ and $\equiv _{T_f}$ is the coarsest amongst the transition congruences of all the DFTs defining $f$.
\end{prf}

\begin{thm}
Let \textbf V be a monoid variety. Let us assume that the membership problem is decidable for \textbf V.
Then there is an algorithm to decide if a function defined by a DFT is definable by a \textbf V-DFT.
\end{thm}

\begin{prf}
This follows easily from Proposition~\ref{prp4}:
We only have to construct the minimal DFT and check if its transition monoid belongs to \textbf V.
\end{prf}

\begin{xmp}
Let us consider the sequential function $g$ which is equal to $f_{ends}$, restricted to $L_{ends}$.
In other words, $g_{ends}$ is defined on words that start and end with an $a$ and replaces each letter of the input by an $a$.
We have:
\begin{itemize}
\item $\widehat g(\varepsilon)=a$.
\item For $w\in a+a\Sigma^\ast a$, $\widehat g(w)=a^{|w|}$.
\item For $w\in a\Sigma^\ast b$, $\widehat g(w)=a^{|w|+1}$.
\item For $w\in b\Sigma^\ast$, the function is not defined.
\end{itemize}
Hence we obtain the following congruence classes:
$[\varepsilon]$, $[a]=[a+a\Sigma^\ast]$,  $[ab]=[a\Sigma^\ast b]$ and $[b]=[b\Sigma^\ast]$.
From this we obtain the minimal DFT of $g$ given in Figure~\ref{figdft}, with states labelled as representatives of congruence classes.

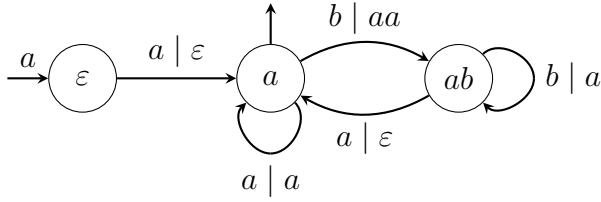
\begin{figure}[t]
\centering
\begin{tikzpicture}
\node[blanc,minimum height=9mm] (e) at (0,0) {$\varepsilon$} ;
\node[blanc,minimum height=9mm] (a) at (2.5,0) {$a$} ;
\node[blanc,minimum height=9mm] (ab) at (5,0) {$ab$} ;
\draw[thick,->,>=stealth] (-1,0) --(e) node[midway,above,sloped] {$a$};
\draw[thick,->,>=stealth] (a) --(2.5,1);
\draw[thick,->,>=stealth] (e)--(a) node[midway,above,sloped] {$a\mid\varepsilon$};
\draw[thick,->,>=stealth] (a)to[bend left](ab);
\node at (3.75,0.8) {$b\mid aa$};
\node at (3.75,-0.8) {$a\mid \varepsilon$};
\draw[thick,->,>=stealth] (ab)to[bend left](a);

\draw[thick] (a) to[out=-45,in=0](2.5,-1) node[below] {$a\mid a$};
\draw[thick,->,>=stealth] (2.5,-1) to[out=180,in=-135](a);
\draw[thick] (ab) to[out=45,in=90](6,0) node[right] {$b\mid a$};
\draw[thick,->,>=stealth] (6,0) to[out=-90,in=-45](ab);

\end{tikzpicture}
\caption{Minimal DFT of $g$.}
\label{figdft}
\end{figure}

\end{xmp}

\section{Determinization of aperiodic transducers}
It has been shown that it is decidable whether an NFT is determinizable or not. In this case we will describe the determinization algorithm given in \cite{bealc02} and show it preserves the aperiodicity of the transducer.
In particular, we have that a subsequential function is definable by an aperiodic DFT if and only if it is definable by an aperiodic NFT, which does not hold for a general variety as can be seen in Section~\ref{app-det}.

Let $T=(Q,I,F,\Delta,i,t)$ be an NFT. We give a construction of $T'=(Q',S_0,F',\delta,i',t')$ equivalent to $T$.
Since we are only interested in the transitions of $T'$, we will not describe the final states nor the terminal output function.

In the obtained transducer $T'$, a state $S$ is a set of pairs $(p,w)\in Q\times \Sigma^\ast$.
Let $X\subseteq \Sigma^\ast$, $X\neq \varnothing$ then we define $\wedge X$ to be the longest common prefix of all the words in $X$.
Let $i'=\wedge\left\{i(q)\mid q\in I \right\}$.
Let $S_0=\left\{(q,w)\mid i(q)=i'w\right\}$.
Then for $S_1$ a state and $\sigma \in \Sigma$, we define $R_2=\left\{ (q,vu)\mid  (p,v)\in S_1\ \mathrm{and}\ p\xrightarrow{\sigma\mid u}q\right\}$.
Let $s=\wedge\left\{w\mid (q,w)\in R_2\right\}$.
Then we define $S_2=\left\{ (q,w)\mid  (q,sw)\in R_2 \right\}$ and there is a transition $S_1\xrightarrow{\sigma\mid s}S_2$ in $T'$.
Let us remark that, by construction, in the underlying automaton of $T'$, $S_1\xrightarrow{u}S_2$ implies that the states of $S_2$ are exactly the states reachable by reading $u$ from a state of $S_1$.

\begin{prp}
Given a subsequential function defined by an \textbf A-NFT, the transducer obtained by determinization is an \textbf A-DFT.
\end{prp}

\begin{prf} 
First let us denote by $u \prec v$ that $u$ is a suffix of $v$. And let $u\approx v$ if $u \prec v$ or $v \prec u$.\\
Let $f$ be a rational function defined by an \textbf A-NFT $T=(Q,I,F,\Delta,i,t)$.
There exists an integer $n$ such that in the underlying automaton of $T$, for all $p,q\in Q$, we have $p\xrightarrow {u^n}q$ if and only if $p\xrightarrow {u^{n+1}}q$.

An automaton is \emph{counter-free} if for any state $q$, any word $u$, any positive integer $k$, $q\xrightarrow{u^k}q \Rightarrow q\xrightarrow{u}q$.
Since $T'$ is deterministic, we only have to show that it is counter-free (see \emph{e.g.} \cite{diekertg08}).

Let $u$ be a word, let $R_0$ be a state of $T'$ and let $k$ be, if it exists, the smallest positive integer such that $R_0\xrightarrow{u^k}R_0$. Let $R_0\xrightarrow{u\mid x_1}R_1\ldots, R_{k-1}\xrightarrow{u\mid x_0}R_{0}$ be the corresponding sequence of transitions in $T'$. We have to show that $k=1$.
Let us remark that all the states of the $R_j$'s must be the same since they are the set of states reachable, in $T$, by $u^n$ from the states of $R_0$. Let $R_j=\left\{(q_1,v_{1,j}),\ldots, (q_m,v_{m,j})\right\}$ be pairwise distinct for $j\in\left\{ 0,\ldots, k-1 \right\}$. We assume that not all $x_j$'s are empty words, otherwise, the conclusion is immediate.

Let $q_{i_1}\xrightarrow {u\mid x_{i_1,i_2}} q_{i_2}$ denote the transitions in $T$. Then by construction of $T'$, we have for any for $j\in\left\{ 0,\ldots, k-1 \right\}$: $$v_{i_1,j}x_{i_1,i_2}=x_{j+1}v_{i_2,j+1}$$
Let $i_0,i_1,\ldots, i_t$ be an index sequence such that $q_{i_l}\xrightarrow uq_{i_{l+1}}$ for $l\in\left\{ 0,\ldots,t-1\right\}$ and $t$ a multiple of $k$.
Then we obtain for any $j,j'\in\left\{ 0,\ldots, k-1 \right\}$:
$$\begin{array}{rll} 
v_{i_0,j}&x_{i_0,i_1}\cdots x_{i_{t-1},i_t}=x_{j+1}\cdots x_{j+t}v_{i_t,j}\\
v_{i_0,j'}&x_{i_0,i_1}\cdots x_{i_{t-1},i_t}=x_{j'+1}\cdots x_{j'+t}v_{i_t,j'}
\end{array}$$
In particular the two words have a common suffix which means that $v_{i_t,j}$ and $v_{i_t,j'}$ have a common suffix.
This suffix can be arbitrarily large since $t$ can be chosen arbitrarily large.
Hence $v_{i,j} \approx v_{i,j'}$ for any $i,j,j'$ (any state $q_i$ is reachable by a sufficiently long sequence).

Let $X_0=x_0\cdots x_{k-1}$. For $j\in\left\{ 1,\ldots, k-1 \right\}$ and let $X_j=x_j\cdots,x_{j-1}$ denote the corresponding rotations of $X_0$. Let us remark that $X_jx_j=x_jX_{j+1}$.
Since $T'$ is finite, there exists an index $i$ such that we have both:
$q_i\xrightarrow{u^t}q_i$ and $q_i\xrightarrow{u^{t+1}}q_i$, and let $i=i_1,i_2,\ldots, i_t,i$ and $i=i'_1,i'_2,\ldots,i'_{t+1},i$ be the corresponding index sequences.
By aperiodicity, we can choose $t=sk$, a multiple of $k$. 
Let $Y=x_{i_1,i_2}\cdots x_{i_t,i_1}$ and let $Y'=x_{i'_1,i'_2}\cdots x_{i'_{t+1},i'_1}$. Then we have for any $j$:

$$\begin{array}{rcl} 
v_{i,j}Y&=&(X_j)^sv_{i,j}\\
v_{i,j+1}Y&=&(X_{j+1})^sv_{i,j+1}\\
v_{i,j}Y'&=&(X_j)^sx_{j}v_{i,j+1}\\
\end{array}$$
Let us assume that $v_{i,j}\prec v_{i,j+1}$ which means that $v_{i,j+1}=wv_{i,j}$.
$$\begin{array}{rclr} 
v_{i,j}Y&=&(X_j)^sv_{i,j}&\quad (\alpha)\\
wv_{i,j}Y&=&(X_{j+1})^swv_{i,j}&\quad(\beta)\\
v_{i,j}Y'&=&(X_j)^sx_{j}wv_{i,j}&\quad(\gamma)\\
\end{array}$$
From $(\alpha)$ and $(\beta)$, we have: $w(X_j)^s=(X_{j+1})^sw$.
And by multiplying by $x_j$ and using $X_jx_j=x_jX_{j+1}$ we have:
$x_jw(X_j)^s=(X_j)^sx_jw$.

From $(\gamma)$, by multiplying by $Y'$ on the right $k$ times we have:
$v_{i,j}(Y')^k=((X_j)^sx_jw)^kv_{i,j}$.
Moreover, since $(Y')^k$ corresponds to a sequence of length a multiple of $k$, we also have:
$v_{i,j}(Y')^k=(X_j)^{ks+1}v_{i,j}$.
Hence $(X_j)^{ks+1}=((X_j)^sx_jw)^k$.
And by combining the two we obtain $(X_j)^{ks+1}=(X_j)^{ks}(x_jw)^k$, hence:
$$X_j=(x_jw)^k$$
This yields

$\begin{array}{rl}
(X_{j+1})^sw&=w(X_j)^s\\
&=w(x_jw)^{ks}\\
&=(wx_j)^{ks}w\\
\end{array}$

Hence
$$ X_{j+1}=(wx_j)^k$$

Let $q_l$ be another state such that $q_l\xrightarrow{u^t}q_l$.
If $v_{l,j}\prec v_{l,j+1}$ then $v_{l,j+1}=w_lv_{l,j}$.
However, it also implies that $(x_jw)^k=(x_jw_l)^k$, hence $w_l=w$.
On the other hand, let us assume $v_{l,j}\succ v_{l,j+1}$ with $w_lv_{l,j+1}=v_{l,j}$.
Let $Y_l$ and $Y'_l$ correspond this time to sequences associated with $q_l\xrightarrow{u^t}q_l$ and $q_l\xrightarrow{u^{t-1}}q_l$, respectively.

$$\begin{array}{rcl} 
w_lv_{l,j+1}Y_l&=&(X_j)^sw_lv_{l,j+1}\\
v_{l,j+1}Y_l&=&(X_{j+1})^sv_{l,j+1}\\
v_{l,j+1}Y'_l&=&(X_{j+1})^{s}x_{j+1}\cdots x_{j-1} w_lv_{l,j+1}\\
\end{array}$$
As above, the first two equations give:
$(X_j)^sw_l=w_l(X_{j+1})^s$

From the third equation we have:
$((X_{j+1})^sx_{j+1}\cdots x_{j-1} w_l)^k=(X_{j+1})^{ks+k-1}$. Let us remark that $(X_{j+1})^sx_{j+1}\cdots x_{j-1}=x_{j+1}\cdots x_{j-1}(X_j)^s$, hence $(X_{j+1})^sx_{j+1}\cdots x_{j-1} w_l=x_{j+1}\cdots x_{j-1}w_l(X_{j+1})^s$. From this we obtain $(x_{j+1}\cdots x_{j-1} w_l)^k=(X_{j+1})^{k-1}$.

Now from above we can write:

$$\begin{array}{rcl}
(x_{j+1}\cdots x_{j-1} w_l)^k&=&((wx_j)^k)^{k-1}\\
x_{j+1}\cdots x_{j-1} w_l&=&(wx_j)^{k-1}\\
wx_jx_{j+1}\cdots x_{j-1} w_l&=&(wx_j)^k\\
wX_jw_l&=&X_{j+1}
\end{array}$$
Hence by length, we obtain $w=w_l=\varepsilon$.

Let $q_l$ be a state reachable from $q_i$ (any state is reachable from such a looping state). Then it can be reached by a sequence of length $t$ and there is a correspond sequence of outputs denoted by $Z_l$ such that:
$$\begin{array}{rcl} 
v_{i,j}Z_l&=&(X_j)^sv_{l,j}\\
wv_{i,j}Z_l&=&(X_{j+1})^sv_{i,j+1}\\
\end{array}$$
Hence $wv_{l,j}=v_{l,j+1}$.

Finally, $w$ is a common prefix to all $v_{i',j+1}$ ($j$ fixed and $i'\in \left\{1,\cdots,m\right\}$) and must by definition be empty. The reasoning is the same if we assume $v_{i,j}\succ v_{i,j+1}$. Hence $v_{i',j}=v_{i',j+1}$ and therefore $k=1$, which means that $T'$ is counter-free and thus is aperiodic.
\end{prf}

\begin{rmk}
In Section~\ref{app-det} we give several examples of varieties \textbf V for which, contrary to the variety of aperiodic monoids, determinization of a transducer does not preserve the variety of the transition monoid.
In fact we even give examples of subsequential functions which are not definable by any \textbf V-DFT yet are definable by a \textbf V-NFT.
\end{rmk}

\begin{xmp}
Let $g$ be defined by the aperiodic NFT of Figure~\ref{detxmp}~(a). Then the DFT obtained by determinization is the one of Figure~\ref{detxmp}~(b).

\begin{figure}[t]
        \centering
        \begin{subfigure}[b]{0.4\textwidth}
\centering
\begin{tikzpicture}
\node[blanc] (q0) at (0,0) {0} ;
\node[blanc] (q1) at (2,1.5) {1} ;
\node[blanc] (q2) at (4,0) {2} ;
\draw[thick,->,>=stealth] (-1,0) --(q0);
\draw[thick,<-,>=stealth] (5,0) --(q2);
\draw[thick,->,>=stealth] (q0)--(q1) node[midway,above,sloped] {$a\mid aa$};
\draw[thick,->,>=stealth] (q0)--(q2) node[midway,below,sloped] {$a\mid a$};
\draw[thick,->,>=stealth] (q1)--(q2) node[midway,above,sloped] {$a\mid\varepsilon$};
\draw[thick] (q1) to[out=45,in=0](2,2.5) node[above] {$a,b\mid a$};
\draw[thick,->,>=stealth] (2,2.5) to[out=180,in=135](q1);
\end{tikzpicture}
\caption{NFT defining $g$.}
        \end{subfigure}
        ~
        \begin{subfigure}[b]{0.4\textwidth}
\centering
\begin{tikzpicture}
\node[rectangle,draw, fill=white] (0) at (0,0) {$0,\varepsilon$} ;
\node[rectangle,draw, fill=white] (12) at (2.5,0) {$\begin{array}{c}1,a\\ 2,\varepsilon\end{array}$} ;
\node[rectangle,draw, fill=white] (1) at (5,0) {$1,\varepsilon$} ;
\draw[thick,->,>=stealth] (-1.1,0) --(0);
\draw[thick,->,>=stealth] (12) --(2.5,1.3);
\draw[thick,->,>=stealth] (0)--(12) node[midway,above,sloped] {$a\mid a$};
\draw[thick,->,>=stealth] (12)to[bend left](1);
\node at (3.75,0.8) {$b\mid aa$};
\node at (3.75,-0.8) {$a\mid \varepsilon$};
\draw[thick,->,>=stealth] (1)to[bend left](12);

\draw[thick] (12) to[out=-65,in=0](2.5,-1.3) node[below] {$a\mid a$};
\draw[thick,->,>=stealth] (2.5,-1.3) to[out=180,in=-115](12);
\draw[thick] (1) to[out=35,in=90](6.1,0) node[right] {$b\mid a$};
\draw[thick,->,>=stealth] (6.1,0) to[out=-90,in=-35](1);

\end{tikzpicture}
\caption{Determinized transducer.}
        \end{subfigure}
        \caption{Determinization of a transducer}
        \label{detxmp}
\end{figure}
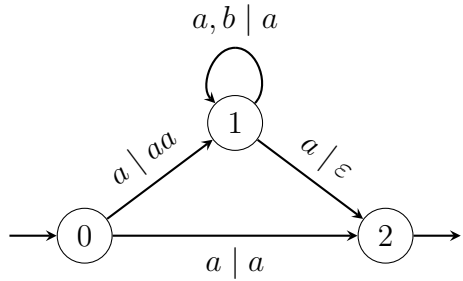
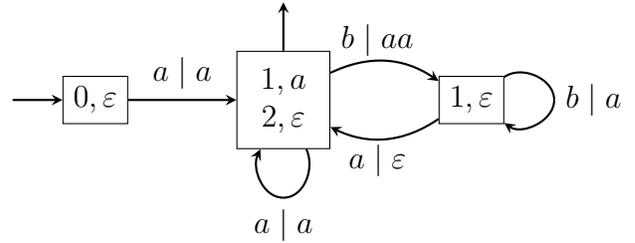

\end{xmp}

\chapter{Algebraic characterization of classes of rational functions}

In this chapter we give a characterization of functions that can be defined by a \textbf V-NFT for a given variety \textbf V. The goal is ultimately to decide if a function can be defined by a logical translation for a logic corresponding to the variety.
The notion of minimal automaton does not transfer to NFTs.
However there is some notion of canonical bimachine, introduced in \cite{reutenauers91}, which is the foundation of the algebraic characterization we give.

\section{V-transducer and V-bimachines}
Since the algebraic characterization deals with bimachines, we need to define a notion of bimachine equivalent to \textbf V-transducers.
A bimachine is called a \emph{\textbf V-bimachine} if both its automata are \textbf V-automata.
We show the following relations between \textbf V-transducers and \textbf V-bimachines.

\begin{prp}
\label{prp3}
Let \textbf V be a monoid variety.
Let $f$ be a rational function.
Then there is a equivalence between the three following properties:
\begin{enumerate}
\item The function $f$ is definable by an unambiguous \textbf V-transducer.
\item The function $f$ is definable by a complete \textbf V-bimachine.
\item The function $f$ is definable by a \textbf V-bimachine and $\mathrm{dom}(f)$ is a \textbf V-language.
\end{enumerate}
Furthermore there is a construction from each one to the others.
\end{prp}

\begin{prf}
$1.\Rightarrow 3.$ follows from the construction given in Section~\ref{tobim}.
Since $\equiv_T$ is a \textbf V congruence, we have that $L$ and $R$ are both \textbf V-automata. Moreover, $\mathrm{dom}(f)$ is a \textbf V-language since it is accepted by the underlying automaton of a \textbf V-NFT.

$3.\Rightarrow 2.$ follows from the construction given in Proposition~\ref{complete}. We can choose the automaton $A$ to be a \textbf V-automaton since the domain is a \textbf V-language. Hence we obtain a complete \textbf V-bimachine.

$2.\Rightarrow 1.$ follows from the construction given in Section~\ref{totrans}. Since the bimachine is complete, the underlying automaton of the obtained transducer is just the product $L\times \tilde R$, with $\tilde R$ being the left automaton obtained by reversing the transitions in $R$. We have that $\equiv_{\tilde R}$ and $\equiv_R$ are equal. Therefore the transducer we obtain is a \textbf V-transducer.
\end{prf}

\begin{rmk}
For some variety \textbf V, there exist functions definable by a \textbf V-bimachine for which the domain is \emph{not} a \textbf V-language.
An example is given in Section~\ref{app-v-dom}.

A question that remains is the following one:
If a function is definable by a \textbf V-NFT, is it definable by some unambiguous \textbf V-NFT ?
In the next part, we will see that the answer is yes in the particular case of aperiodic transducers.
\end{rmk}

\section{Canonical bimachine}

In this part we give the construction, for a given rational function, of a canonical bimachine that is described in \cite{reutenauers91} and several properties of this bimachine. It is canonical in the sense that it is not machine dependant and only depends on the function itself.
Although, we give the construction of the bimachine, we refer the reader to the original paper for a proof of correctness.

In order to define the bimachine we define a left congruence from which we define the right automaton. Then from this we define the right congruence which yields the left automaton.

\subsection{Left congruence}
Let $f$ be a rational function  on $\Sigma$.

Let the \emph{left distance} between two words $u,v$ be $\lVert u,v \rVert=|u|+|v|-2|u\wedge v|$ where $u\wedge v$ is the longest common prefix of $u$ and $v$.

We define the \emph{left syntactic congruence} of $f$ (which is a left congruence) by $ u\sim_{R_0} v$ if:
$$\begin{array}{c}
\forall w\in \Sigma^\ast,\  wu\in \mathrm{dom}(f) \Leftrightarrow wv\in \mathrm{dom}(f)\\
\text{and}\\
\sup_{w\in \Sigma^\ast}\left\{\lVert f(wu),f(wv)\rVert \right\}<\infty
\end{array}$$

Let $\sim_R $ be a left congruence finer than $\sim_{R_0}$. Let $R=\Sigma^\ast / \sim_R$, $R_0=\Sigma^\ast / \sim_{R_0}$ and let $r_0=[\varepsilon]$ be the initial element of $R$.

For $r,r'\in R$, $\sigma\in \Sigma$, if $r'\xleftarrow{\sigma}r$ then for convenience we denote by $r'=\sigma r$ the action of $\Sigma$ on $R$.

For this fixed $R$ we construct a bimachine $B^R=(L^R,R,\omega^R,\lambda^R,\rho^R)$ which defines $f$ according to \cite{reutenauers91}. By taking $R_0$ as a right automaton we get a completely canonical bimachine $B^0=(L^0,R_0,\omega^0,\lambda^0,\rho^0)$.
Since we are only interested in the automata of the bimachine, we will not give the descriptions of $\omega^R,$ $\lambda^R$ and $\rho^R$, and once again we refer the reader to \cite{reutenauers91}.

\begin{rmk}
The construction is completely symmetrical: we can define similarly a right congruence: $u\sim_{L_0}v$ if the supremum of the \emph{right distance} between $f(uw)$ and $f(vw)$ over all words $w$ is less than infinity.
Then for any right congruence $\sim_L$ finer than $\sim_{L_0}$, we can define $B^L=(L,R^L,\omega^L,\lambda^L,\rho^L)$ defining $f$. We denote $R^{L_0}$ by $R^0$ to simplify.
\end{rmk}

First we show that the automaton $R_0$ (as well as $L_0$) has some minimal property, \emph{i.e.} its transition congruence is coarser than the transition congruence of any transducer defining $f$.

\begin{prp}
\label{prp2}
Let $f$ be a rational function defined by some \textbf V-NFT.
Then the automaton $R_0=(\Sigma^\ast / \sim_{R_0},[\varepsilon],\delta)$,
where $\delta$ is the natural left action on $\Sigma^\ast / \sim_{R_0}$, is a \textbf V-automaton.
\end{prp}
\begin{prf}
Let $T=(Q,I,F,\Delta,i,t)$ be a \textbf V-NFT defining $f$.
According to Proposition~\ref{v-congruence}, it suffices to prove ${\sim_T}\sqsubseteq{\equiv_{R_0}}$.
Let $w\in \Sigma^\ast$ be two words, let $m_1\equiv_Tm_2$ and let $q_0\xrightarrow{w\mid x}p\xrightarrow{m_1\mid y}q_f$ be an accepting run of $T$ on $wm_1$.
Then there is an accepting run $q_0\xrightarrow{w\mid x}p\xrightarrow{m_2\mid y'}q_f$.
Thus the distance $\lVert f(wm_1),f(wm_2)\lVert\leq k(|m_1|+|m_2|+1)$ with $k$ being the maximum length of an output in $T$.
Hence $m_1\sim_{R_0} m_2$.
\end{prf}

\subsection{Right congruence}

Now that we have defined the right automaton of the bimachine, we need to define the left one. In order to do that we define this time a right congruence.

Let $r\in R$, we define the function
$$\widehat f_r(u)=\bigwedge \left\{f(uv)\mid  vr_0=r \right\}$$

These functions are similar to the prefix function defined for minimizing subsequential functions in Section~\ref{mindft}. However this time we need a finite family of prefix functions.
It is proven in \cite{reutenauers91} that $\widehat f_{\sigma r}(u)$ is a prefix of $\widehat f_r(u\sigma)$ for $\sigma\in \Sigma, u\in \Sigma^\ast$ and $r\in R$.
Then for $u,v\in \Sigma^\ast$ we say $u\sim_{L^R} v$ if for any $\sigma\in \Sigma, w\in \Sigma^\ast$ and $r\in R$ we have when it is defined both:
$$\begin{array}{rcl}
\widehat f_{\sigma r}(uw)^{-1}\widehat f_{r}(uw\sigma) &=& \widehat f_{\sigma r}(vw)^{-1}\widehat f_{r}(vw\sigma)\\
\widehat f_{r_0}(uw)^{-1} f(uw) &=& \widehat f_{r_0}(vw)^{-1} f(vw)
\end{array} $$

This right congruence has finite index, according to \cite{reutenauers91}, thus we obtain the left automaton of the bimachine $B^R$ defining $f$.

For the purpose of this report, we want to study the algebraic properties of $L^R$ and for this we need to show that the functions $\widehat f_r$, $r\in R$ maintain some algebraic properties of $T$ (given below in Proposition~\ref{prp1} and Corollary~\ref{cor1}) and for this we give the construction of a family of transducers defining them.
Let us give a construction of a transducer defining $\widehat f_r$, for $r\in R$.
According to \cite{reutenauers91} there exists $L_r=\left\{ v_1,\ldots v_m\right\}$ finite such that $\widehat f_r(u)=\bigwedge \left\{f(uv)\mid  v\in L_r \right\}$.

Let $T\times R=(Q\times R,I\times R,F\times\left\{ r_0\right\},\Delta',i',t')$ be the transducer that does exactly as $T$ does except that it guesses in a non-deterministic fashion the congruence class $\sim_R$ of the word read, and checks by taking the transitions of $R$ backwards to reach $r_0$.
In particular this transducer defines the function $f$, the states are just refined to also give the class in $R$ of the suffix left to read.

Now from this transducer we want to build $T_r=(Q_R,I_R,F_r,\Delta_R,i_R,t_r)$ defining $\widehat f_r$ with $i_R=\wedge\left\{i(q)\mid q\in I\right\}$.
The construction will be similar to the determinization algorithm, except that only states with same class in $R$ will be merged.
The algorithm will end since for $w_1{\sim_R}w_2$ we have by definition the left distance between $f(ww_1)$ and $f(ww_2)$ that is bounded.

Each state of $T_r$ is composed of a set of pairs $(q,w)\in Q\times \Sigma^\ast$ and a class $s\in R$.
For $i_R=\wedge\left\{i(q)\mid q\in I\right\}$ and $s\in R$, each $(\left\{(q,w)\mid i(q)=i_Rw\right\},s)$ belongs to $I_R$.
Let $(P,s)$ be a state of $Q_R$, let $\sigma\in \Sigma$ and $s'\in \Sigma$ such that $s=\sigma s'$.
Then we can determine a set
$$S=\left\{ (q',wu)\mid  \exists (q,w)\in P\ \mathrm{and}\ q\xrightarrow{\sigma\mid u}q'\right\}$$
Let $v$ be the longest common prefix of the words of $S$ then we define
$$P'=\left\{ (q',w')\mid   (q',vw')\in S\right\}$$
There is then a transition $(P,s)\xrightarrow{\sigma\mid v}(P',s')$.
We define $F_r=\left\{ (P,s)\in Q_R\mid  s=r \right\}$.
Let $(P,r)\in F_r$ with $P=\left\{(q_1,w_1),\ldots, (q_m,w_m) \right\}$. For any $v\in L_r$ there exists $j_v$ such that for some $q_{f_v}\in F$, $q_{j_v}\xrightarrow{v\mid u_v}q_{f_v}$.
Then we define $t_r((P,r))=\bigwedge_{} \left\{ w_{j_v}u_v\mid v\in L_r\right\}$.

By construction, $T_r$ defines $\widehat f_r$, furthermore this construction has the advantage of yielding an unambiguous transducer: for a given word, there is only one sequence of states in $R$ corresponding to it and since all states in $Q\times R$ following this sequence are merged, there is only one possible accepting run.

Let $T_R=(Q_R,I_R,\tilde\Delta_R)$ denote the (unique) underlying automaton associated with the $T_r$'s (when not taking final states into account).

For the following, we also need to show that $f$ is definable by a NFT with $T_R$ as its underlying automaton. We define $T_f=(Q_R,I_R,F_f,\Delta_R,i_R,t_f)$ with the final states $F_f\left\{(P,r_0)\mid \exists (p,w)\in P \text{ and } p\in F\right\}$ and the terminal function $T_f(P,r_0)=wt(p)$ for $(p,w)\in P \text{ and } p\in F$.

The construction of $T_R$ allows us to claim the following property and its corollary.
\begin{prp}
\label{prp1}
Let $f$ be a rational function defined by some \textbf V-NFT.
If the automaton $T_R$ is a \textbf V-automaton, then so is the automaton $L^R=(\Sigma^\ast / \sim_{L^R},[\varepsilon],\delta)$, where $\delta$ is the natural right action on $\Sigma^\ast / \sim_L$.
\end{prp}
\begin{prf}
Let $T=(Q,I,F,\Delta,i,t)$ be a \textbf V-NFT defining $f$.
According to Proposition~\ref{v-congruence} it suffices to show ${\equiv_{T^R}}\sqsubseteq{\sim_{L^R}}$.
Let $m_1\equiv_{T_R}m_2$, $r\in R$ and $w \in \Sigma^\ast$.
Let $(P_0,m_1w\sigma r)\xrightarrow{m_1\mid u_1}(P_1,w\sigma r)\xrightarrow {w\mid v} (P_2,\sigma r)\xrightarrow {\sigma \mid v'} (P_3, r)$ and 
$(P_0,m_1w\sigma r)\xrightarrow{m_2\mid u_2}(P_1,w\sigma r)\xrightarrow {w\mid v} (P_2,\sigma r)\xrightarrow {\sigma \mid v'} (P_3, r)$ be successful runs of $T_r$ (if they exist). Then we have:
$$\begin{array}{l}
\widehat f_r(m_1w\sigma)=u_1vv't_r(P_3,r) \\
\widehat f_r(m_2w\sigma)=u_2vv't_r(P_3,r)\\
\widehat f_{\sigma r}(m_1w)=u_1vt_{\sigma r}(P_2,\sigma r)\\
\widehat f_{\sigma r}(m_2w)=u_2vt_{\sigma r}(P_2,\sigma r)\\
\end{array}$$
By definition $t_{\sigma r}(P_2,\sigma r)$ is indeed a prefix of $v't_r(P_3,r)$, hence $\widehat f_{w r}(m_1)^{-1}\widehat f_r(m_1w)=\widehat f_{w r}(m_2)^{-1}\widehat f_r(m_2w)$.

We have $m_1\in \mathrm{dom}(f)$ if and only if $m_2\in \mathrm{dom}(f)$.
Let $(P_0,m_1wr_0)\xrightarrow{m_1\mid u_1}(P_1,wr_0)\xrightarrow {w\mid v} (P_2,r)$ and 
$(P_0,m_1wr_0)\xrightarrow{m_2\mid u_2}(P_1,wr_0)\xrightarrow {w\mid v} (P_2,r)$ be successful runs of $T_f$ (if they exist).
Then we have:
$$\begin{array}{l}
f(m_1w)=u_1vt_f(P_2,r_0) \\
f(m_2w)=u_2vt_f(P_2,r_0)\\
\widehat f_{r_0}(m_1w)=u_1vt_{r_0}(P_2,r_0)\\
\widehat f_{r_0}(m_2w)=u_2vt_{r_0}(P_2,r_0)\\
\end{array}$$
By definition $t_{r_0}(P_2,r_0)$ is indeed a prefix of $t_f(P_2,r_0)$ since $\varepsilon \in r_0$, hence $\widehat f_{r_0}(m_1)^{-1}f(m_1)=\widehat f_{r_0}(m_2)^{-1}f(m_2)$, and $m_1\sim_{L^R}m_2$.
\end{prf}

\begin{cor}
\label{cor1}
Let \textbf V be a monoid variety stable by determinization of transducers.
Let $f$ be definable by a \textbf V-transducer.
Then $B^0$ is a \textbf V-bimachine.
\end{cor}
\begin{prf}
According to Proposition~\ref{prp2}, $R_0$ is a \textbf V-automaton.
Hence by construction $T_{R_0}$ is a \textbf V-automaton since \textbf V is stable by determinization.
Hence, according to Proposition~\ref{prp1}, $L^0$ is also a \textbf V-automaton.
\end{prf}

\begin{xmp}
We consider the function $f_{ends}$ and one can show that the automaton $R_0$ is equal to the one of Example~\ref{xmp-bim}, \emph{i.e.} $u\sim_{R_0} v$ if they have the same last letter.
The automaton $T_R$ we obtain for the NFT defining $f_{ends}$ is given in Figure~\ref{fig-tr}.

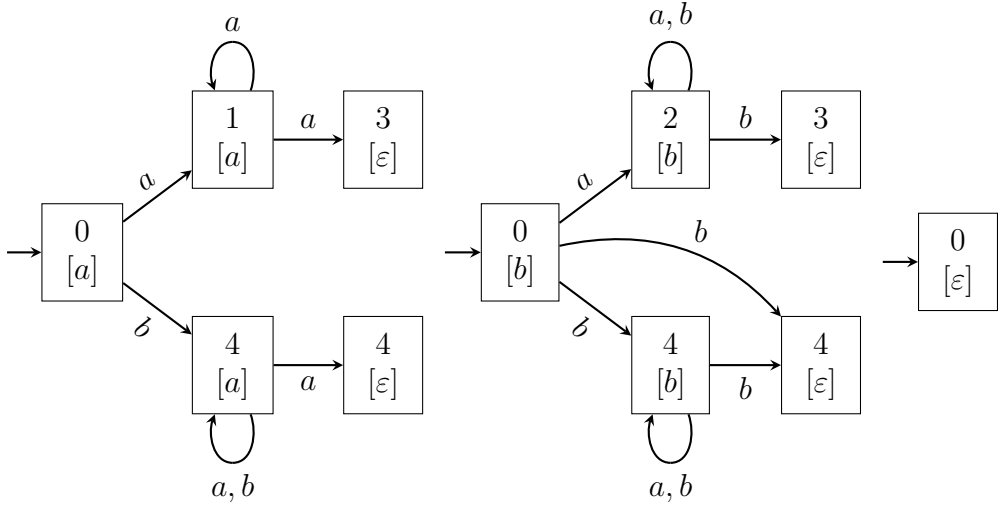
\begin{figure}[t]
\centering
\begin{tikzpicture}

\node[rectangle,draw, fill=white] (0a) at (0,0) {$\begin{array}{c} 0 \\ \left[ a\right] \end{array}$} ;

\node[rectangle,draw, fill=white] (1a) at (2,1.5) {$\begin{array}{c} 1 \\ \left[ a\right] \end{array}$} ;
\node[rectangle,draw, fill=white] (3e) at (4,1.5) {$\begin{array}{c} 3 \\ \left[ \varepsilon\right] \end{array}$} ;
\node[rectangle,draw, fill=white] (4a) at (2,-1.5) {$\begin{array}{c} 4 \\ \left[ a\right] \end{array}$} ;
\node[rectangle,draw, fill=white] (4e) at (4,-1.5) {$\begin{array}{c} 4 \\ \left[ \varepsilon\right] \end{array}$} ;

\draw[thick,<-,>=stealth] (0a) -- ++(-1,0);
\draw[thick,->,>=stealth] (0a)--(1a) node[midway,above,sloped] {$a$};
\draw[thick,->,>=stealth] (0a)--(4a) node[midway,below,sloped] {$b$};
\draw[thick,->,>=stealth] (1a)--(3e) node[midway,above,sloped] {$a$};
\draw[thick,->,>=stealth] (4a)--(4e) node[midway,below,sloped] {$a$};

\draw[thick] (1a) to[out=70,in=0] (2,2.8) node[above] {$a$};
\draw[thick,->,>=stealth] (2,2.8) to[out=180,in=110] (1a);
\draw[thick] (4a) to[out=-70,in=0] (2,-2.8) node[below] {$a,b$};
\draw[thick,->,>=stealth] (2,-2.8) to[out=180,in=-110] (4a);
\end{tikzpicture}~
\begin{tikzpicture}

\node[rectangle,draw, fill=white] (0a) at (0,0) {$\begin{array}{c} 0 \\ \left[ b\right] \end{array}$} ;

\node[rectangle,draw, fill=white] (1a) at (2,1.5) {$\begin{array}{c} 2 \\ \left[ b\right] \end{array}$} ;
\node[rectangle,draw, fill=white] (3e) at (4,1.5) {$\begin{array}{c} 3 \\ \left[ \varepsilon\right] \end{array}$} ;
\node[rectangle,draw, fill=white] (4a) at (2,-1.5) {$\begin{array}{c} 4 \\ \left[ b\right] \end{array}$} ;
\node[rectangle,draw, fill=white] (4e) at (4,-1.5) {$\begin{array}{c} 4 \\ \left[ \varepsilon\right] \end{array}$} ;

\draw[thick,<-,>=stealth] (0a) -- ++(-1,0);
\draw[thick,->,>=stealth] (0a)--(1a) node[midway,above,sloped] {$a$};
\draw[thick,->,>=stealth] (0a)--(4a) node[midway,below,sloped] {$b$};
\draw[thick,->,>=stealth] (1a)--(3e) node[midway,above,sloped] {$b$};
\draw[thick,->,>=stealth] (4a)--(4e) node[midway,below,sloped] {$b$};
\draw[thick,->,>=stealth] (0a) to[bend left](4e);
\node[above] at (2.4,0) {$b$};

\draw[thick] (1a) to[out=70,in=0] (2,2.8) node[above] {$a,b$};
\draw[thick,->,>=stealth] (2,2.8) to[out=180,in=110] (1a);
\draw[thick] (4a) to[out=-70,in=0] (2,-2.8) node[below] {$a,b$};
\draw[thick,->,>=stealth] (2,-2.8) to[out=180,in=-110] (4a);
\end{tikzpicture}~
\begin{tikzpicture}

\node[rectangle,draw, fill=white] (0a) at (0,0) {$\begin{array}{c} 0 \\ \left[ \varepsilon\right] \end{array}$} ;
\draw[thick,<-,>=stealth] (0a) -- ++(-1,0);
\node at (0,-3.2){};
\end{tikzpicture}

\caption{Automaton $T_R$.}
\label{fig-tr}
\end{figure}

\end{xmp}

\subsection{Aperiodic case}

In Chapter~\ref{chap4} we have proven that aperiodic transducers are stable by determinization. Hence we can decide if a given function is definable by an aperiodic transducer, and claim the following Theorem:

\begin{thm}
There is an algorithm to decide if a function defined by a transducer is definable by an FO-translation.
\end{thm}

\begin{prf}
Let us assume $f$ is definable by an \textbf A-NFT. Then its domain is an \textbf A-language and, according to Corollary~\ref{cor1}, the canonical bimachine is an \textbf A-bimachine.
Let us assume that $f$ is definable by an FO-translation.
According to Proposition~\ref{t-to-b}, we can construct a complete aperiodic bimachine defining $f$.
Then according to Proposition~\ref{prp3} we can construct an \textbf A-NFT defining $f$.
Hence to decide if a rational function given by a NFT is definable by a \textbf A-NFT, we need to:
\begin{enumerate}
\item Check that the domain is an aperiodic language.
\item Construct the canonical bimachine $B^0$.
\item If the domain is an aperiodic language and $B^0$ is aperiodic, from Proposition~\ref{prp3} we can build a complete aperiodic bimachine $B$ defining the function.
\item From $B$, according to Proposition~\ref{prpbt}, we can construct an FO-translation defining the function.
\end{enumerate}
\end{prf}

\begin{rmk}
In particular, a function is definable by an aperiodic transducer if and only if it is definable by an unambiguous transducer, since the construction from bimachine to transducer always yields an unambiguous transducer.
\end{rmk}

\section{General variety}

In this section we try to generalize the previous result to any variety.
We show that if a function is definable by an unambiguous \textbf V-NFT, then there exists a bimachine amongst a finite set of "canonical" bimachines that is a \textbf V-bimachine.
We first show the result on total functions and then extend it to the general case.

Let $A$, $A'$ be two deterministic left (resp. right) automata.
In order to simplify the notations, we will write $A\sqsubseteq {A'}$ and say that the automaton $A$ is finer than $A'$ whenever ${\sim_A}\sqsubseteq{\sim_{A'}}$.

\subsection{Total functions}
The construction is a finer version of the previous one and heavily uses the following Theorem of \cite{reutenauers91}.
\begin{thm}
\label{thm1}
Let $f$ be a total rational function, let $R$ be a right automaton finer than $R_0$ and let $L^R$ be the left automaton of the bimachine $B^R$.
Then we say that $L^R$ is \emph{minimal for  $R$} in the following sense:
for any bimachine defining $f$ with $R$ as its right automaton and $L'$ as its left automaton, we have $${L'} \sqsubseteq {L^R}$$
\end{thm}

We also use the following Proposition which holds for total functions only.
\begin{prp}
Let $B=(L,R,\omega,\lambda,\rho)$ be a bimachine defining some total function $f$.
Then we have both ${L} \sqsubseteq {L_0}$ and ${R} \sqsubseteq {R_0}$.
\end{prp}
\begin{prf}
Let us show ${R} \sqsubseteq {R_0}$ since the proof is the same for $L_0$.
Let $u,v\in \Sigma^\ast$ such that $u\sim_R v$.
It means that $r_0\xrightarrow u r$ and $r_0\xrightarrow v r$ for some $r\in R$.
Let $w=\sigma_1\ldots\sigma_n\in \Sigma^n$ and let $l_0\xrightarrow {\sigma_1} l_1 \ldots l_{n-1}\xrightarrow {\sigma_n} l_n$ be the execution of $w$ on $L$.
Similarly, let $r_n\xrightarrow {\sigma_1} r_{n-1} \ldots r_{1}\xrightarrow {\sigma_n} r$ be the execution of $w$ on $R$ from $r$.
Then
$$\begin{array}{rcl}
\lVert f(wu),f(wv)\rVert &=&
\lVert \lambda(r_n)\omega(l_0,w,r)\omega(l_n,u,r_0)\rho(l),\lambda(r_n)\omega(l_0,w,r)\omega(l_n,v,r_0)\rho(l')\rVert\\
&=&\lVert \omega(l_n,u,r_0)\rho(l),\omega(l_n,v,r_0)\rho(l')\rVert\\
&\leq& k(|u|+|v|+1)
\end{array}$$
with $k$ being the maximum size of an output in $B$.
Hence $u\sim_{R_0}v$.
\end{prf}

In the following Proposition, we show that $L^0$ is maximal amongst the set of minimal left automata, but before we need the Lemma:

\begin{lem}
\label{lem-finer}
Let $B$ be a bimachine with $L$ as its left automaton and $R$ as its right automaton. Let $L'$ be finer than $L$ and $R'$ be finer than $R$. Then there is a bimachine $B'$ with $L'$ and $R'$ as its automata and equivalent to $B$.
\end{lem}

\begin{prf}
Since $L'$ and $R'$ are finer than $L$ and $R$, respectively, the bimachine $B'$ can behave as $B$ by disregarding the finer information on its states.
\end{prf}

\begin{prp}
Let $f$ be a total rational function. Let $R$ be a right automaton finer than $R_0$. Let $L^R$ be the left automaton of the bimachine  $B^R$.
Then $${L^0} \sqsubseteq {L^R}$$
\end{prp}

\begin{prf}
This follows from Theorem~\ref{thm1}: Since ${R} \sqsubseteq {R_0}$, there is, according to Lemma~\ref{lem-finer}, a bimachine defining $f$ with $L^0$ and $R$ as its automata and we have ${L^0} \sqsubseteq {L^R}$.
\end{prf}

For a complete function $f$ we have proven that $L^0$ is maximal amongst the set of minimal left automata.
This means that if we want to show that $f$ is definable by some bimachine with a \textbf V-automaton as its left automaton, we only have to check if the finite set of automaton coarser than $L^0$ (and finer than $L_0$) contains a \textbf V-automaton.

If there exists a \textbf V-automaton coarser than $L^0$ and finer than $L_0$, then we define $L_\mathbf{V}$ as the finest left \textbf V-automaton such that ${L^0} \sqsubseteq {L_\mathbf{V}} \sqsubseteq {L_0}$.
Finally we define $R^\mathbf V=R^{L_\mathbf V}$ as the right automaton of the bimachine $B^{L_\mathbf V}$.

Before moving on, let us show the following lemma:
\begin{lem}
\label{lem_v_c}
Let ${\sim_1}\sqsubseteq {\sim_2}$ be two right (resp. left) congruences.
Let $\equiv_1$, $\equiv_2$ denote the transition congruences of the left (resp. right) automata associated with $\sim_1$ and $\sim_2$, respectively. Then ${\equiv_1}$ is finer than $ {\equiv_2}$.
\end{lem}

\begin{prf}
Let us assume that $\sim_1$, $\sim_2$ are left congruences, the proof being completely symmetrical.
We have that $u\equiv_i v$ if and only if for all $w\in \Sigma^\ast$, $wu\sim_i wv$.
Hence we can easily see that $u\equiv_1 v \Rightarrow u\equiv_2 v$.
\end{prf}

We can now claim the desired property:

\begin{prp}
Let \textbf V be a monoid variety.
Let $f$ be a total rational function.
The function $f$ is definable by an unambiguous \textbf V-NFT if and only if $R^\mathbf V$ is defined and is a \textbf V-automaton.
\end{prp}

\begin{prf}
Let $f$ be a total function, let us assume that there exists a \textbf V-NFT defining $f$. According to Proposition~\ref{prp3} there exist a \textbf V-bimachine $B=(L,R,\omega,\lambda,\rho)$ defining $f$.
Let $L^R$ be the left automaton of the bimachine $B^R$ with $R$ as its right automaton.
Then, according to Theorem~\ref{thm1}, we have ${L} \sqsubseteq {L^R}$, hence from Lemma~\ref{lem_v_c}, $L^R$ is a \textbf V-automaton.
We have ${L^0} \sqsubseteq {L_\mathbf{V}} \sqsubseteq {L^R} \sqsubseteq {L_0}$.
Let $R'=R^{(L^R)}$ be the right automaton of the bimachine $B^{(L^R)}$ with $L^R$ as its left automaton.
We have, again according to Theorem~\ref{thm1}, ${R} \sqsubseteq {R'}$, hence $R'$ is also a \textbf V-automaton.
We have ${R^0} \sqsubseteq {R'} \sqsubseteq {R^\mathbf V} \sqsubseteq {R_0}$. Hence according to Lemma~\ref{lem_v_c}, $R^\mathbf V$ is also a \textbf V-automaton.

Conversely, we have in particular $B^{L_\mathbf V}$ a \textbf V-bimachine defining the total function $f$ hence, according to Proposition~\ref{prp3}, we can obtain a \textbf V-NFT defining $f$.
\end{prf}

\begin{rmk}
In particular, total functions are completely characterized by a finite set of canonical bimachines in the following sense:
Let $f$ be a complete function defined by a bimachine $B$ with $L$ and $R$ as its left and right automata. Then there exists a bimachine $B_{min}$ defining $f$ with $L_{min}$ and $R_{min}$ as its left and right automata with $L\sqsubseteq L_{min}$ and $R\sqsubseteq R_{min}$ such that:
$L^0 \sqsubseteq L_{min} \sqsubseteq L_0$ and $R^0 \sqsubseteq R_{min} \sqsubseteq R_0$.
\end{rmk}

\subsection{Partial functions}

Now that we have the result for total function, we need to extend it to partial function.
Let $f:\Sigma^\ast\rightarrow \Sigma^\ast$ be a (partial) rational function.
We define the completion of $f$:
$$\begin{array}{cccc}\bar f:& \Sigma^\ast &\rightarrow & (\Sigma\uplus\left\{\bot \right\})^\ast \\
&u & \mapsto & \left\{\begin{array}{l l} f(u) & \mathrm{if} \ u\in \mathrm{dom}(f) \\ \bot & \mathrm{otherwise}\end{array}  \right.
\end{array} $$

\begin{lem}
\label{lem_comp}
Let \textbf V be a monoid variety.
Let $f$ be a rational function such that $\mathrm{dom}(f)$ is a \textbf V-language.
The function $f$ is definable by an unambiguous \textbf V-NFT if and only if $\bar f$ is definable by an unambiguous \textbf V-NFT.
\end{lem}

\begin{prf}
Let us assume $f$ is definable by an unambiguous \textbf V-NFT $T$.
Then there is a \textbf V-automaton recognizing the complement of $\mathrm{dom}(f)$.
We can easily see that the complement of a \textbf V-language is a \textbf V-language: If $M$ is a monoid in \textbf V such that $\mathrm{dom}(f)=\mu^{-1}(P)$ for some morphism $\mu$ and $P\subseteq M$, then $\Sigma^\ast\backslash \mathrm{dom}(f)=\mu^{-1}(M\backslash P)$.
By adding to $T$ a deterministic transducer that maps any word not in the domain of $f$ to $\bot$, we obtain an unambiguous \textbf V-NFT defining $\bar f$.

Conversely, let us assume $\bar f$ is definable by a \textbf V-NFT. Since $\bar f$ is total, according to Proposition~\ref{prp3} $f$ is definable by a \textbf V-bimachine. Then by restricting the output function to values not equal to $\bot$, we obtain a \textbf V-bimachine defining $f$. According to Proposition~\ref{prp3}, again, $f$ is definable by a \textbf V-NFT.
\end{prf}

\begin{rmk}
We have reduced the problem of deciding whether a function is definable by an unambiguous \textbf V-NFT to the membership problem of \textbf V.
\end{rmk}

\begin{thm}
\label{v-nft}
Let \textbf V be a monoid variety. Let us assume that the membership problem is decidable for \textbf V.
Then there is an algorithm to decide if a function defined by an NFT is definable by an unambiguous \textbf V-NFT.
\end{thm}
\begin{proof}
According to Proposition~\ref{prp3} and Lemma~\ref{lem_comp}, $f$ is definable by an unambiguous \textbf V-NFT iff $\bar f$ is definable by a \textbf V-NFT and the domain of $f$ is a \textbf V-language iff $\bar f$ is definable by a \textbf V-bimachine and the domain of $f$ is a \textbf V-language.

The construction goes as follows:
First we check that the domain of $f$ is a \textbf V-language.
We construct a NFT defining $\bar f$.
Then we construct the bimachine $B^{L_\mathbf V}$, if it exists, and check if $R^\mathbf V$ is a \textbf V-automaton.
If it is the case we use the construction of Lemma~\ref{lem_comp} to obtain a \textbf V-bimachine defining $f$.
Finally, from the construction of Proposition~\ref{prp3} we obtain an unambiguous \textbf V-NFT defining $f$

\end{proof}

\chapter{Remarks, examples and counter-examples}

\section{Monoid varieties}
The monoid varieties we consider most often are given in Figure~\ref{table},
 which is taken from \cite{dartois15}. We also give the corresponding logical and language characterizations.
The equalities are given with the symbol $\omega$ which can be seen as the \emph{idempotent power} of a finite monoid $M$, \emph{i.e.} the smallest integer $n$ such that any element to the $n$ is idempotent.
These equalities, called \emph{profinite equalities}, also characterize the monoid varieties. For a proper definition we refer the reader to \cite{pin97}.

\begin{rmk}
Since one can determine the idempotent power of a finite monoid, we can see that the membership problem is decidable for varieties defined by a finite set of profinite equalities. This is true for the varieties listed in Figure~\ref{table} and in particular for the variety of aperiodic monoids.
\end{rmk}
\begin{figure}[t]
\centering

\begin{tabular}{|c|c|c|c|}

\hline
Variety&Equations&Logic&Languages\\
\hline
\hline
Finite Monoids & - & MSO & Rational\\
\hline
Commutative \textbf{Com} & $xy=yx$ & MSO$[=]$ & Commutative\\
\hline
Aperiodic \textbf A & $x^\omega=x^{\omega+1}$ & FO & Star-free\\
\hline
$\mathcal D$-Aperiodic \textbf{DA} & $(xyz)^\omega y(xyz)^\omega=(xyz)^\omega$ & FO$^2=\Delta_2$ & Unambiguous\\
\hline
\textbf{J}$_1$ & $x=x^2,xy=yx$ & FO$^1$ & -\\
\hline
$\mathcal J$-trivial \textbf J & $y(xy)^\omega=(xy)^\omega=(xy)^\omega x$ & $\mathcal B \Sigma_1$ & \begin{tabular}{c}Piecewise\\testable\end{tabular}\\
\hline
\end{tabular}
\caption{Monoid varieties}
\label{table}
\end{figure}

\section{Determinization}
\label{app-det}
Here we see that, in general, determinization of transducers does not preserve the equalities of the transition monoid.

\begin{xmp}

Let us consider \textbf I the variety of idempotent monoids, \emph{i.e.} satisfying $x=x^2$.
Let $\Sigma=\left\{a,b\right\}$.
We define the function:
$$f: a^n\mapsto\left\{ \begin{array}{c} a\ \mathrm{if}\ n=1\\ b\ \mathrm{if}\ n>1   \end{array} \right.$$
The NFT of Figure~\ref{det1} defines $f$ and verifies $x=x^2$.

\begin{figure}[t]
\centering

\begin{tikzpicture}
\node[blanc] (q0) at (0,0) {0} ;
\node[blanc] (q1) at (2,1) {1} ;
\node[blanc] (q2) at (2,-1) {2} ;
\draw[thick,->,>=stealth] (-1,0) --(q0);
\draw[thick,->,>=stealth] (q1) -- (3,1);
\draw[thick,->,>=stealth] (q0)--(q1) node[midway,above,sloped] {$a\mid a$};
\draw[thick,->,>=stealth] (q0)--(q2) node[midway,below,sloped] {$a\mid b$};
\draw[thick,->,>=stealth] (q2)--(q1) node[midway,right] {$a\mid \varepsilon$};
\draw[thick] (q2) to[out=-45,in=0](2,-2) node[below right] {$a\mid \varepsilon$};
\draw[thick,->,>=stealth] (2,-2) to[out=180,in=-135](q2);
\end{tikzpicture}
\caption{\textbf I-NFT.}
\label{det1}
\end{figure}

Let us assume that $T'$ is an \textbf I-DFT defining $f$. Then we can decompose $f(a)=ut$ and $f(a^2)=uvt$, with $q_0\xrightarrow{a\mid u}q\xrightarrow{a\mid v}q$ and $t$ is the terminal output corresponding to $q$. Then $u=t=f(a)=\varepsilon$, which yields the contradiction.
\end{xmp}

\begin{xmp}
There is a sequential function definable by a \textbf{Com}-NFT that is not definable by any \textbf{Com}-DFT.
Let $\Sigma=\left\{a,b,c\right\}$.
We define the function $f$:
$$ \begin{array}{rcl}
ab &\mapsto& a\\
aab & \mapsto& a\\
aba &\mapsto& ab\\
ba &\mapsto& c \\
baa &\mapsto& c\\

\end{array}$$
The NFT of Figure~\ref{det2} defines $f$ and verifies $xy=yx$.

\begin{figure}[t]
\centering

\begin{tikzpicture}
\node[blanc] (q0) at (0,0) {0} ;
\node[blanc] (q1) at (3,-1) {1} ;
\node[blanc] (q2) at (3,1) {2} ;
\node[blanc] (q3) at (6,0) {3} ;
\node[blanc] (q4) at (6,2) {4} ;
\node[blanc] (q5) at (9,1) {5} ;
\node[blanc] (q6) at (-1,-2) {6} ;
\node[blanc] (q7) at (1,-2) {7} ;
\node[blanc] (q8) at (0,-4) {8} ;
\draw[thick,->,>=stealth] (-1,0) --(q0);
\draw[thick,->,>=stealth] (q5) -- (10,1);
\draw[thick,->,>=stealth] (q8) -- (1,-4);
\draw[thick,->,>=stealth] (q0)--(q2) node[midway,above] {$a\mid a$};
\draw[thick,->,>=stealth] (q2)--(q4) node[midway,above] {$a\mid \varepsilon$};
\draw[thick,->,>=stealth] (q1)--(q3) node[midway,above] {$a\mid \varepsilon$};
\draw[thick,->,>=stealth] (q3)--(q5) node[midway,above] {$a\mid \varepsilon$};
\draw[thick,->,>=stealth] (q0)--(q6) node[midway,left] {$a\mid a$};
\draw[thick,->,>=stealth] (q7)--(q8) node[midway,right] {$a\mid \varepsilon$};

\draw[thick,->,>=stealth] (q0)--(q1) node[midway,above] {$b\mid c$};
\draw[thick,->,>=stealth] (q2)--(q3) node[midway,above] {$b\mid b$};
\draw[thick,->,>=stealth] (q4)--(q5) node[midway,above] {$b\mid \varepsilon$};
\draw[thick,->,>=stealth] (q0)--(q7) node[midway,right] {$b\mid c$};
\draw[thick,->,>=stealth] (q6)--(q8) node[midway,left] {$b\mid \varepsilon$};

\end{tikzpicture}
\caption{\textbf{Com}-NFT.}
\label{det2}
\end{figure}

Let us assume that $T'$ is a DFT defining $f$ and verifying $xy=yx$. Then we can decompose $f(ab)=u_1t_1$,$f(ba)=u_2t_1$, $f(aba)=u_1vt_2$ and $f(baa)=u_2vt_2$ with $q_0\xrightarrow{ab\mid u_1}q_1\xrightarrow{a\mid v}q_2$, $q_0\xrightarrow{ba\mid u_2}q_1$ and $t_1$, $t_2$ are the terminal outputs corresponding to $q_1$ and $q_2$, respectively. Since $f(aba)$ and $f(baa)$ have no letter in common, we have $v=t_2=\varepsilon$. Hence $u_1=ab$, and we have a contradiction.
\end{xmp}

\begin{xmp}
There is a sequential function definable by a \textbf J-NFT that is not definable by any \textbf J-DFT.
Let $\Sigma=\left\{a,a_0,a_1,b,b_0,b_1,c,d\right\}$.
The \textbf J-NFT of Figure~\ref{det4} defines a subsequential function $f$.

\begin{figure}[t]
\centering

\begin{tikzpicture}
\node[blanc] (q0) at (0,0) {0} ;
\node[blanc] (q1) at (3,0) {1} ;
\node[blanc] (q2) at (-60:3) {2} ;
\draw[thick,->,>=stealth] (0,1) --(q0);
\draw[thick,->,>=stealth] (q2) -- ++(1,0);
\draw[thick,->,>=stealth] (q0)--(q1) node[midway,above] {$a\mid a_1$};
\draw[thick,->,>=stealth] (q0)--(q1) node[midway,below] {$b\mid b_1$};
\draw[thick,->,>=stealth] (q0)--(q2) node[midway,left] {$c\mid \varepsilon$};
\draw[thick,->,>=stealth] (q1)--(q2) node[midway,right] {$d\mid \varepsilon$};
\draw[thick,->,>=stealth] (q0) to[out=135,in=90] (-1,0) to[out=-90,in=-135](q0);

\node[above left] at (-1,0) {$a\mid a_0$}; 
\node[below left] at (-1,0) {$b\mid b_0$}; 

\end{tikzpicture}
\caption{\textbf J-NFT.}
\label{det4}
\end{figure}

Let us assume that $T'$ is a \textbf J-DFT defining $f$. Then we can decompose $f((ab)^\omega c)=ut_c$, $f((ab)^\omega ac)=u\alpha t_c$ and $f((ab)^\omega ad)=u\alpha t_d$ with $q_0\xrightarrow{(ab)^\omega\mid u}q\xrightarrow{a\mid \alpha}q$. Since $f((ab)^\omega c)$ and $f((ab)^\omega ac)$ have no common suffix, we have $t_c=\epsilon$. Hence we obtain that $f((ab)^\omega ac)$ is a prefix of $f((ab)^\omega ad)$ which gives the contradiction.
\end{xmp}

\section{V-domain}
\label{app-v-dom}
Let $\Sigma=\left\{a,b\right\}$.
We give an example of a \textbf{Com}-bimachine $B=(L,R,\omega,\lambda,\rho)$ such that its domain is not a commutative language.
The right automaton is the complete automaton with one state. The left automaton is given in Figure~\ref{v-bim}.
The functions $\lambda$ and $\rho$ are constant and equal to $\varepsilon$.
We define $\omega(l_0,a,r_0)=\omega(l_1,b,r_0)=\varepsilon$ and $\omega$ is not defined otherwise.
Then the domain of the function is equal to $\left\{ab\right\}$ which is not a commutative language.
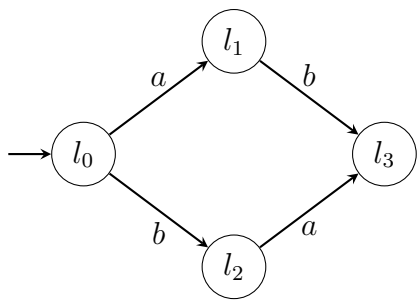
\begin{figure}[t]
\centering

\begin{tikzpicture}
\node[blanc] (l0) at (0,0) {$l_0$} ;
\node[blanc] (l1) at (2,1.5) {$l_1$} ;
\node[blanc] (l2) at (2,-1.5) {$l_2$} ;
\node[blanc] (l3) at (4,0) {$l_3$} ;
\draw[thick,<-,>=stealth] (l0) -- ++(-1,0);
\draw[thick,->,>=stealth] (l0)--(l1) node[midway,above] {$a$};
\draw[thick,->,>=stealth] (l0)--(l2) node[midway,below] {$b$};
\draw[thick,->,>=stealth] (l1)--(l3) node[midway,above] {$b$};
\draw[thick,->,>=stealth] (l2)--(l3) node[midway,below] {$a$};

\end{tikzpicture}
\caption{Left automaton of the bimachine $B$.}
\label{v-bim}
\end{figure}

\section{Disambiguation}

\begin{xmp}
It has been shown that any rational function is definable by some unambiguous transducer. Moreover there is for transducers a disambiguation algorithm which yields an unambiguous transducer with an exponential number of states.
Let $\Sigma=\left\{a,b\right\}$.
We show in Figure~\ref{dis} an example of an aperiodic automaton for which the algorithm described in \cite{filiotg12} never yields an aperiodic automaton.

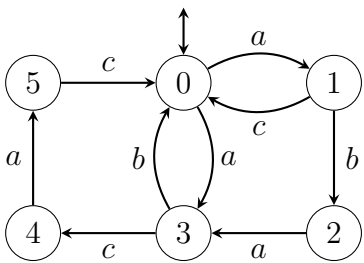
\begin{figure}[t]
\centering

\begin{tikzpicture}
\node[blanc] (q0) at (0,2) {0} ;
\node[blanc] (q1) at (2,2) {1} ;
\node[blanc] (q2) at (2,0) {2} ;
\node[blanc] (q3) at (0,0) {3} ;
\node[blanc] (q4) at (-2,0) {4} ;
\node[blanc] (q5) at (-2,2) {5} ;

\draw[thick,<->,>=stealth] (0,3) --(q0);
\draw[thick,->,>=stealth] (q3) --(q4) node[midway, below] {$c$};
\draw[thick,->,>=stealth] (q3) to[bend left] (q0) ;
\draw[thick,->,>=stealth] (q4) --(q5) node[midway, left] {$a$};
\draw[thick,->,>=stealth] (q5) --(q0) node[midway, above] {$c$};
\draw[thick,->,>=stealth] (q0) to[bend left] (q3);
\draw[thick,->,>=stealth] (q0) to[bend left] (q1);
\draw[thick,->,>=stealth] (q1) to[bend left] (q0);
\draw[thick,->,>=stealth] (q1) --(q2) node[midway, right] {$b$};
\draw[thick,->,>=stealth] (q2) --(q3) node[midway, below] {$a$};
\node at (-0.6,1){$b$};
\node at (0.6,1){$a$};
\node at (1,2.6){$a$};
\node at (1,1.4){$c$};

\end{tikzpicture}
\caption{Aperiodic NFA.}

\label{dis}
\end{figure}

As you can see this automaton is aperiodic, yet not counter-free.
For state $0$ there are two possible transitions reading letter $a$. If we choose $(0,a,1)<_\delta (0,a,3)$ then in the disambiguated automaton we have
$$(0,\varnothing)\xrightarrow{(ab)^{k}}(2,\varnothing)$$
If and only if $k$ is odd. Similarly if $(0,a,3)<_\delta (0,a,1)$ then
$$(0,\varnothing)\xrightarrow{(ac)^{k}}(4,\varnothing)$$
If and only if $k$ is odd. In both cases the obtained automaton is periodic.
\end{xmp}

\chapter*{Conclusion}
Using techniques inspired by \cite{choffrut03,reutenauers91}, we have been able to extend decidability results for varieties of rational languages to varieties of rational functions.
In the deterministic case the notion of minimal DFT gives a natural way to decide if a subsequential functions belongs to a given variety.
In the non-deterministic case we need to consider not one minimal machine, but a finite set of canonical bimachines and if one of these machines satisfies the desired algebraic property, then the function it defines belongs to the corresponding variety.

Some of the questions that remain are for instance the complexity of the decision procedures described in this report.
The question of the equivalence between \textbf V-NFT and unambiguous \textbf V-NFT for a general variety also remains unanswered. 

We foresee several interesting directions in which the techniques of this report can be taken.
The first one would be to try to extend the results to the more general class of regular functions, defined by two-way transducers \cite{engelfrieth01} and streaming string transducers \cite{alurc10}, or some intermediate class.
A second possible direction is the generalization to relations, with probably some condition of finite value, for instance by bounding the ambiguity of the transducers.
Another extension would be to consider infinite words and see for instance if some minimization is possible for DFTs on infinite words.
There is also the domain of weighted automata for which common techniques are shared with transducers \cite{filiotgr14}.

\bibliographystyle{alpha}
\bibliography{biblio}

\begin{thebibliography}{MSTV06}

\bibitem[A{\v C}10]{alurc10}
Rajeev Alur and Pavol {\v C}ern{\'{y}}.
\newblock Expressiveness of streaming string transducers.
\newblock In {\em {IARCS} Annual Conference on Foundations of Software
  Technology and Theoretical Computer Science, {FSTTCS} 2010, December 15-18,
  2010, Chennai, India}, pages 1--12, 2010.

\bibitem[BB79]{berstelb79}
Jean Berstel and Luc Boasson.
\newblock Transductions and context-free languages.
\newblock {\em Ed. Teubner}, pages 1--278, 1979.

\bibitem[BC02]{bealc02}
Marie{-}Pierre B{\'{e}}al and Olivier Carton.
\newblock Determinization of transducers over finite and infinite words.
\newblock {\em Theor. Comput. Sci.}, 289(1):225--251, 2002.

\bibitem[B{\"u}c60]{buchi60}
Julius~Richard B{\"u}chi.
\newblock Weak second-order arithmetic and finite automata.
\newblock {\em Mathematical Logic Quarterly}, 6(1-6):66--92, 1960.

\bibitem[CD15]{cartond15}
Olivier Carton and Luc Dartois.
\newblock Aperiodic two-way transducers and {FO-}transductions.
\newblock {\em CSL}, 2015.

\bibitem[Cho03]{choffrut03}
Christian Choffrut.
\newblock Minimizing subsequential transducers: a survey.
\newblock {\em Theor. Comput. Sci.}, 292(1):131--143, 2003.

\bibitem[Dar15]{dartois15}
Luc Dartois.
\newblock {\em M{\'e}thodes alg{\'e}briques pour la th{\'e}orie des automates}.
\newblock PhD thesis, Universit{\'e} Paris-Diderot, {\'E}cole doctorale Paris
  Centre, 2015.

\bibitem[DG08]{diekertg08}
Volker Diekert and Paul Gastin.
\newblock First-order definable languages.
\newblock In {\em Logic and Automata: History and Perspectives [in Honor of
  Wolfgang Thomas].}, pages 261--306, 2008.

\bibitem[EF95]{ebbinghausF95}
Heinz{-}Dieter Ebbinghaus and J{\"{o}}rg Flum.
\newblock {\em Finite model theory}.
\newblock Perspectives in Mathematical Logic. Springer, 1995.

\bibitem[EH01]{engelfrieth01}
Joost Engelfriet and Hendrik~Jan Hoogeboom.
\newblock {MSO} definable string transductions and two-way finite-state
  transducers.
\newblock {\em {ACM} Trans. Comput. Log.}, 2(2):216--254, 2001.

\bibitem[Elg61]{elgot61}
Calvin~Creston Elgot.
\newblock Decision problems of finite automata design and related arithmetics.
\newblock {\em Transactions of the American Mathematical Society}, 98(1):pp.
  21--51, 1961.

\bibitem[ES76]{eilenbergs76}
Samuel Eilenberg and Marcel{-}Paul Sch{\"u}tzenberger.
\newblock On pseudovarieties.
\newblock {\em Advances in Mathematics}, 19(3):413 -- 418, 1976.

\bibitem[FGR12]{filiotg12}
Emmanuel Filiot, Raffaella Gentilini, and Jean{-}Fran{\c{c}}ois Raskin.
\newblock Quantitative languages defined by functional automata.
\newblock In {\em {CONCUR} 2012 - Concurrency Theory - 23rd International
  Conference, {CONCUR} 2012, Newcastle upon Tyne, UK, September 4-7, 2012.
  Proceedings}, pages 132--146, 2012.

\bibitem[FGR14]{filiotgr14}
Emmanuel Filiot, Raffaella Gentilini, and Jean{-}Fran{\c{c}}ois Raskin.
\newblock Finite-valued weighted automata.
\newblock In {\em 34th International Conference on Foundation of Software
  Technology and Theoretical Computer Science, {FSTTCS} 2014, December 15-17,
  2014, New Delhi, India}, pages 133--145, 2014.

\bibitem[Fil15]{filiot15}
Emmanuel Filiot.
\newblock Logic-automata connections for transformations.
\newblock In {\em Logic and Its Applications - 6th Indian Conference, {ICLA}
  2015, Mumbai, India, January 8-10, 2015. Proceedings}, pages 30--57, 2015.

\bibitem[FKT14]{filiotkt14}
Emmanuel Filiot, Shankara~Narayanan Krishna, and Ashutosh Trivedi.
\newblock First-order definable string transformations.
\newblock {\em CoRR}, abs/1406.7824, 2014.

\bibitem[MP71]{mcnaughtonp71}
Robert McNaughton and Seymour Papert.
\newblock {\em Counter{-}free automata}.
\newblock M.I.T. Press Cambridge, Mass, 1971.

\bibitem[MSTV06]{mckenzieSTV06}
Pierre McKenzie, Thomas Schwentick, Denis Th{\'{e}}rien, and Heribert Vollmer.
\newblock The many faces of a translation.
\newblock {\em J. Comput. Syst. Sci.}, 72(1):163--179, 2006.

\bibitem[Ner58]{nerode58}
A.~Nerode.
\newblock Linear automaton transformations.
\newblock {\em Proceedings of the American Mathematical Society}, 9(4):pp.
  541--544, 1958.

\bibitem[Pin97]{pin97}
Jean{-}{\'E}ric Pin.
\newblock Syntactic semigroups.
\newblock In {\em Word, Language, Grammar}, volume~1 of {\em Handbook of Formal
  Languages}, chapter~10, pages 679--738. Springer, 1997.

\bibitem[RS91]{reutenauers91}
Christophe Reutenauer and Marcel{-}Paul Sch{\"{u}}tzenberger.
\newblock Minimization of rational word functions.
\newblock {\em {SIAM} J. Comput.}, 20(4):669--685, 1991.

\bibitem[RS95]{reutenauers95}
Christophe Reutenauer and Marcel{-}Paul Sch{\"u}tzenberger.
\newblock Vari{\'e}t{\'e}s et fonctions rationnelles.
\newblock {\em Theoretical Computer Science}, 145(1–2):229 -- 240, 1995.

\bibitem[Sak09]{sakarovitch09}
Jacques Sakarovitch.
\newblock {\em Elements of Automata Theory}.
\newblock Cambridge University Press, 2009.

\bibitem[Sch61]{schutzenberger61}
Marcel{-}Paul Sch{\"{u}}tzenberger.
\newblock A remark on finite transducers.
\newblock {\em Information and Control}, 4(2-3):185--196, 1961.

\bibitem[Sch65]{schutzenberger65}
Marcel{-}Paul Sch{\"{u}}tzenberger.
\newblock On finite monoids having only trivial subgroups.
\newblock {\em Information and Control}, 8(2):190--194, 1965.

\bibitem[Str96]{straubing96}
Howard Straubing.
\newblock Finite models, automata, and circuit complexity.
\newblock In {\em Descriptive Complexity and Finite Models, Proceedings of a
  {DIMACS} Workshop, January 14-17, 1996, Princeton University}, pages 63--96,
  1996.

\bibitem[Tra61]{trakhtenbrot61}
Boris~Avraamovich Trakhtenbrot.
\newblock Finite automata and logic of monadic predicates.
\newblock {\em Doklady Akademii Nauk SSSR}, 140:326--329, 1961.
\newblock In Russian.

\end{thebibliography}

\end{document}